\documentclass[final]{IEEEtran}

\IEEEoverridecommandlockouts
\usepackage{cite}
\usepackage{amsmath,amssymb,amsfonts,amsthm}
\usepackage{algorithmic}
\usepackage{graphicx}
\usepackage{textcomp}
\usepackage{xcolor}
\usepackage{url}
\usepackage{lipsum}
\usepackage{multirow}
\usepackage{circuitikz}
\usepackage{afterpage}
\usepackage{bbm}
\usepackage{dsfont}
\usepackage[ruled,lined,boxed]{algorithm2e}
\usepackage{mathtools, nccmath}
\usepackage{bm}
\usepackage{threeparttable}
\usepackage{svg}
\usepackage{csquotes}
\usepackage[T1]{fontenc}
\usepackage{makecell}
\usepackage{tikz}
\usepackage{tikz-3dplot}
\usepackage{pgfplots}
\pgfplotsset{compat=1.15}
\usepackage{subcaption}
\usepackage{caption}
\usepackage{siunitx}
\usepackage{cuted}
\usepackage{setspace}
\renewcommand{\baselinestretch}{0.9685}

\makeatletter
\let\orig@subsection\subsection
\def\subsection{\vspace{-4pt}\orig@subsection}
\makeatother

\DeclareMathAlphabet{\mathbcal}{OMS}{cmsy}{b}{n}
\def\BibTeX{{\rm B\kern-.05em{\sc i\kern-.025em b}\kern-.08em
		T\kern-.1667em\lower.7ex\hbox{E}\kern-.125emX}}
 
\newtheorem{prop}{Proposition}

\newtheoremstyle{iremark}
{\topsep}   
{\topsep}   
{\upshape}  
{0pt}       
{\itshape}  
{:}         
{5pt plus 1pt minus 1pt} 
{\thmname{#1}\thmnumber{ \itshape#2}\thmnote{ (#3)}} 
\theoremstyle{iremark}
\newtheorem{remark}{Remark}

\allowdisplaybreaks
\mathchardef\mhyphen="2D
\makeatletter 
\let\myorg@bibitem\bibitem
\def\bibitem#1#2\par{%
	\@ifundefined{bibitem@#1}{%
		\myorg@bibitem{#1}#2\par
	}{%
		\begingroup
		\color{\csname bibitem@#1\endcsname}%
		\myorg@bibitem{#1}#2\par
		\endgroup
	}%
}
\makeatother

\makeatletter
\def\@IEEEtitleabstractindextextfont{%
    \fontsize{10pt}{10pt}\bfseries\selectfont 
}
\makeatother

\begin{document}
{\title{Spatially Reconfigurable Antenna Systems for 6G: EM-based Channel Modeling, Measurements, and Orientation Design\\
}

 \author{
 			Chen~Xu,~\IEEEmembership{Graduate Student Member,~IEEE},
 			Yizhu Yan,~\IEEEmembership{Graduate Student Member,~IEEE},
 			Gengbo Wu,~\IEEEmembership{Senior Member,~IEEE},
 			Shu Sun,~\IEEEmembership{Senior Member,~IEEE},
 			and Xianghao~Yu,~\IEEEmembership{Senior Member,~IEEE}
	\thanks{
		C. Xu, Y. Yan, G. Wu, and X. Yu are with the Department of Electrical Engineering, City University of Hong Kong, Hong Kong (E-mail: cxu297-c@my.cityu.edu.hk; yizhuyan3-c@my.cityu.edu.hk; bogwu2@cityu.edu.hk; alex.yu@cityu.edu.hk).
		
		Shu Sun is with the School of Information Science and Electronic Engineering, Shanghai Jiao Tong University, Shanghai 200240, China (E-mail: shusun@sjtu.edu.cn).
		
		This paper was presented in part at the IEEE 103rd Vehicular Technology Conference (VTC Spring), Nice, France, June 2026~\cite{xu2026generalembasedchannelmodel}.
		
	}
  
 }

\maketitle

\begin{abstract}

Spatially reconfigurable antenna systems (SRASs) are recognized as a key physical-layer technology for sixth-generation (6G) systems.
By dynamically adjusting each antenna element's spatial configuration, e.g., position and orientation, SRASs can revamp favorable channel conditions for reliable high-rate data transmission.
However, in widely adopted channel models, antennas are typically modeled as ideal isotropic radiators, and the vectorial nature of electromagnetic (EM) propagation is neglected.
This oversimplified model precludes full exploitation of the degrees of freedom offered by SRASs for performance enhancement.
To address this issue, in this paper, by leveraging the theoretical framework of spherical vector wave expansion, we develop an EM-based channel model tailored for SRAS-enabled multiple-input multiple-output (MIMO) systems.
The proposed EM-based channel model is applicable to antennas with arbitrary structures and intrinsically accounts for the vectorial nature of EM propagation, thereby enabling accurate characterization of EM effects such as polarization mismatch on channel gain.
\textcolor{black}{Full-wave simulations and experimental measurements are conducted, and the results show excellent agreement with theoretical predictions.}
\textcolor{black}{Simulation results also reveal that antenna orientation exerts a more pronounced influence on the achievable rate than antenna displacement.}
\textcolor{black}{Therefore, building upon the derived channel model, a manifold optimization method is proposed to maximize the sum-rate of an SRAS-enabled multiuser-MIMO system by optimizing antenna orientations.}
Simulation results demonstrate that the proposed scheme improves the sum-rate by up to 16.6\% and 19.3\% compared to systems employing movable antennas and conventional fixed antennas, respectively.

\end{abstract}

\begin{IEEEkeywords}
	Channel model, manifold optimization, spatially reconfigurable antenna system, spherical vector wave expansion.
\end{IEEEkeywords}


\section{Introduction}

Driven by emerging application scenarios such as vehicle-to-everything networks, massive machine-type communications, and integrated space-air-ground-sea networks, sixth-generation (6G) communication systems are expected to provide terabit-per-second connectivity and ultra-wide-area coverage \cite{Song2025}.
However, current multiple-input multiple-output (MIMO)-based wireless networks largely overlook the spatial degrees of freedom (DoFs) inherent in individual antenna elements, thereby limiting their capacity to meet these stringent requirements.
Recent advances in antenna technology have spurred significant interest in spatially reconfigurable antenna systems (SRASs)~\cite{6474484}.
Unlike conventional fixed antenna arrays, SRASs consist of multiple arbitrarily positioned and oriented antenna elements.
By leveraging the spatial reconfigurability of each antenna, SRASs can shape favorable channel conditions to enable reliable high-rate data transmission, and are thus recognized as a key physical-layer enabling technology for future 6G systems.

Previous works have extensively investigated various implementations of SRASs, such as fluid antennas (FAs)~\cite{9264694,10740058,9650760,9131873}, movable antennas (MAs)~\cite{liu2025nfddd,10318061,10848372,liu2025ma}, and rotatable antennas (RAs) \cite{11134688,zheng2025rotatab,Yang2017,11142311}.
Specifically, by leveraging the fluidity of liquid metals, FAs achieve spatial reconfigurability through the activation of feeding ports at different positions within a line space.
In contrast, MAs and RAs primarily rely on mechanically tunable feeding structures to enable physical translation or rotation of the antenna element.
Although existing works have thoroughly investigated the design methodologies and performance gains of SRAS-enabled systems, they still predominantly rely on classical channel models, which were originally developed for fixed antenna arrays.
In particular, classical channel models assume that antennas are ideal isotropic radiators, and thus fail to characterize the radiation properties of arbitrarily oriented and positioned antennas in SRASs.
Moreover, electromagnetic (EM) waves are modeled as scalar plane waves in classical channel models.
This neglect of the vectorial nature of EM propagation prevents current channel models from capturing the profound impact of EM characteristics on the channel gain.
Consequently, these modeling limitations hinder the exploitation of the DoFs offered by antenna spatial configuration, thereby severely constraining the achievable performance of SRAS-based communication systems.

To establish more accurate and comprehensive channel models, recent efforts have shifted toward EM-based channel modeling, which seeks to construct channel models from a fundamental EM perspective~\cite{11006094,9110848,11258091,10500425,10500751,5159555}.
The authors of \cite{11006094} first modeled fading channels as random fields constrained to satisfy the \emph{scalar} wave equations, thereby obtaining a physically consistent representation.
Following this line of insight, several works have sought to extend channel models to a \emph{vectorial} formulation to fully capture key EM effects, such as polarization mismatch caused by misalignment in antenna orientation.
For instance, the authors of \cite{11258091} proposed an EM-based channel model based on the closed-form expression of the vector radiation field.
Although this model intrinsically captures vectorial EM effects, it is restricted to scenarios where both the transceivers employ half-wavelength dipoles, thereby significantly limiting its applicability to more general antenna configurations.
To enhance the generality of channel modeling, the channel models based on dyadic Green's functions have been proposed in \cite{10500425,10500751} and applied to extremely large-scale MIMO and holographic MIMO systems.
In principle, the model is applicable to arbitrary antenna types.
However, it requires prior knowledge of the antenna's surface current distribution, which is inaccessible in practice, hindering its application in real-world scenarios.
Another theoretical framework, known as the spherical vector wave expansion (SVWE), can be employed to analyze the EM propagation between antennas of arbitrary physical structure.
In contrast to Green's function-based methods, SVWE requires only sampled values of the EM field on a grid, which are readily measurable in practical engineering setups.
Once these samples are acquired, SVWE can analytically characterize the EM propagation between transmit and receive antennas.
Owing to this, SVWE has been widely adopted in applications such as antenna measurements~\cite{1143727} and EM scattering analysis~\cite{alma990026142230205776}.
The authors of \cite{5159555} first employed SVWE for channel modeling and demonstrated its potential for capturing EM propagation effects.
However, the model in \cite{5159555} neglected the impact of the transmit antenna's type, position, and orientation on the channel gain.
These simplifications, while reducing analytical complexity, render the model incapable of being employed in SRAS-enabled systems.

Hence, all these limitations of existing works motivate us to develop an EM-consistent channel model which can be employed to comprehensively characterize the channels of general SRASs, thereby enabling a spatial configurability design that fully exploits the DoFs of SRASs to enhance the communication performance. 
This paper investigates the EM-based channel model and the orientation design method tailored for SRAS-enabled MIMO systems.
The main contributions are summarized as follows:
\begin{itemize}
	\item We decompose the channel modeling into three fundamental EM processes: 1) the radiation process, which maps the transmitted signal to the radiated EM field; 2) the propagation process, which transforms the radiated field into incident EM field at the receiver; 3) the reception process, which converts the incident field into the received signal. By leveraging the SVWE framework, we derive exact analytical expressions for these EM processes in the spherical mode domain.
	\item Based on the analysis of the EM processes, we derive a closed-form expression for the channel gain between arbitrarily positioned and oriented antennas in free space, and establish a general EM-based channel model for SRAS-enabled MIMO systems. The proposed channel model is applicable to antennas with arbitrary structures. Moreover, it intrinsically accounts for the vectorial nature of EM propagation, thereby enabling accurate characterization of the impact of EM effects, such as polarization mismatch, on channel gain.
	\item Building upon the derived EM-based channel model, we investigate the sum-rate maximization problem in an SRAS-enabled multiuser (MU)-MIMO system, where the orientation of each transmit antenna can be freely adjusted. Then, a Riemannian gradient ascent method on the special orthogonal group $\mathrm{SO}(3)$ is proposed to find a locally optimal solution to this highly non-convex optimization problem.
	\item Using three-dimensional (3D) printing technology, we fabricated an antenna rotator to enable precise control over the spatial orientation of the antenna. Based on this setup, we measured the channel gain of an SRAS composed of horn antennas in a microwave anechoic chamber. The measurement results show excellent agreement with theoretical predictions, thereby validating the accuracy of the proposed channel model.
\end{itemize}

{\it Notations:} 
$\mathbb{R}$ and $\mathbb{C}$ represent the sets of real numbers and complex numbers, respectively.
$\delta_{i,j}$ denotes the Kronecker delta, which equals $1$ if $i=j$ and $0$ otherwise.
$\mathrm{Re}(z)$ and $\mathrm{Im}(z)$ represent the real and imaginary parts of a complex number $z$, respectively. 
The imaginary unit is represented as $\jmath$ such that $\jmath^2=-1$.
Matrices are denoted by bold uppercase letters (e.g., $\mathbf{A}$).
Vectors are represented by bold lowercase letters (e.g., $\mathbf{a}$).
Scalars are denoted by normal font (e.g., $N$). 
$\left[ \mathbf{A} \right]_{i,j}$ denotes the $(i,j)$-th entry of matrix $\mathbf{A}$, and $\mathbf{a}_{i}$ denotes the $i$-th entry of vector $\mathbf{a}$.
$\left(\cdot\right)^T$, $\left(\cdot\right)^*$, and $\left(\cdot\right)^H$ stand for the transpose, conjugate, and Hermitian transpose of a matrix, respectively.
$\mathrm{Tr}(\mathbf{\cdot})$ indicates the trace of a matrix. 
$\mathbf{I}_{N}$ represents an $N \times N$ identity matrix.
The Frobenius norm is written as $\left\| \cdot\right\|_F$.

\section{Analysis of Radiation, Propagation, and Reception Processes}\label{sec:EM_char_analysis}

Consider an SRAS-enabled MIMO system comprising  $N_{\mathrm{t}}$ transmit and $N_{\mathrm{r}}$ receive antennas.
Denote the transmit signal vector as $\mathbf{v} \in \mathbb{C}^{N_{\mathrm{t}} \times 1}$.
The corresponding receive signal vector $\mathbf{w} \in \mathbb{C}^{N_{\mathrm{r}} \times 1}$ is then given by
\begin{equation}
	\begin{aligned}
  		\mathbf{w} = \mathbf{H} \mathbf{v} + \mathbf{z},
  	\end{aligned}
\end{equation}
where $\mathbf{H} \in \mathbb{C}^{N_{\mathrm{r}} \times N_\mathrm{t}}$ is the channel matrix and $\mathbf{z} \in \mathbb{C}^{N_{\mathrm{r}} \times 1}$ represents the additive white Gaussian noise (AWGN) vector whose entries are independent and identically distributed (i.i.d.) complex Gaussian variables with zero mean and variance $\sigma_\mathrm{n}^2$.

\begin{figure}[t]
	\centering
	\hspace{4mm}\includegraphics[trim=10 1 0 11, clip, width=0.5\textwidth]{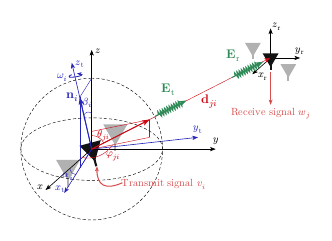}
	\caption{Illustration of the channel modeling between the $i$-th transmit antenna and the $j$-th receive antenna.}
	\label{fig:MIMO_setting_illu} 
	\vspace{-3mm}  
\end{figure}

As shown in Fig.~\ref{fig:MIMO_setting_illu}, two local coordinate systems $(x_\mathrm{t},y_\mathrm{t},z_\mathrm{t})$ and $(x_\mathrm{r},y_\mathrm{r},z_\mathrm{r})$ are centered at the $i$-th transmit and $j$-th receive antennas, respectively, with their $z$-axes aligned with the antenna normal directions.
A global coordinate system $(x,y,z)$ shares its origin with the transmit antenna and its $z$-axis is parallel to the $z_\mathrm{r}$-axis.
In the global frame $(x,y,z)$, the $i$-th transmit antenna's spatial orientation is defined by its normal vector $\mathbf{n}_i$, parameterized by azimuth angle $\alpha_i$, elevation angle $\beta_i$, and spin angle $\omega_i$ (rotation about $\mathbf{n}_i$).
The relative position from the $i$-th transmit antenna to the $j$-th receive antenna is described by the displacement vector $\mathbf{d}_{ji}$, characterized by the azimuth angle $\varphi_{ji}$ and elevation angle $\theta_{ji}$.

Let $h_{ji}=[\mathbf{H}]_{j,i}$ denote the channel gain from the $i$-th transmit to the $j$-th receive antenna.
As illustrated in Fig.~\ref{fig:MIMO_setting_illu}, the transformation from $v_i$ to $w_j$ decomposes into three fundamental EM processes:
A)~radiation: conversion of $v_i$ into the radiated electric field $\mathbf{E}_{\mathrm{t}}$ by the $i$-th transmit antenna;
B)~propagation: transport of $\mathbf{E}_{\mathrm{t}}$ through free space, yielding the incident field $\mathbf{E}_{\mathrm{r}}$ at the receiver;
C)~reception: conversion of $\mathbf{E}_{\mathrm{r}}$ into the received signal $w_j$ at the $j$-th receive port.
These processes are analyzed in the following three subsections, respectively, within the SVWE framework.

\subsection{Radiation Characterization: From $v_i$ to $\mathbf{E}_\mathrm{t}$}

The vector Helmholtz equation governs any time-harmonic EM field in source-free regions. 
As its general solution in spherical coordinates, the spherical vector wave functions (SVWFs) constitute a complete and orthogonal basis for representing arbitrary time-harmonic EM fields. 
By expanding the field as a linear superposition of SVWFs, the SVWE provides a mathematically complete and physically exact framework for EM field characterization \cite{Stratton2015}.
Unlike plane- and cylindrical-wave expansions, SVWE leverages the inherent symmetry of SVWFs, and thereby enables a natural and efficient characterization of how antenna rotations and translations shape the EM field.
This makes SVWE particularly well-suited for analyzing channel characteristics in SRAS-enabled systems.

According to the SVWE framework, when observed in the local spherical coordinate system $(r_\mathrm{t},\theta_\mathrm{t},\varphi_\mathrm{t})$ associated with $(x_\mathrm{t},y_\mathrm{t},z_\mathrm{t})$, the radiated electric field $\mathbf{E}_{\mathrm{t}}$ can be decomposed as
\begin{equation}\label{eq:elec_field_as_sum_spherical_func}
	\begin{aligned}
  		\mathbf{E}_{\mathrm{t}}(r_\mathrm{t},\theta_\mathrm{t},\varphi_\mathrm{t})=\frac{\kappa v_i}{\sqrt{\eta}}  \sum_{s=1}^{2}\sum_{n=1}^{N}\sum_{m=-n}^{n}   T_{smn}\mathbf{F}_{smn}^ {(3)} (r_\mathrm{t},\theta_\mathrm{t},\varphi_\mathrm{t}),
  	\end{aligned}
\end{equation}
where $\kappa$ and $\eta$ denote the wavenumber and the wave impedance of the medium, respectively. 
The truncation index $N$ is a constant determined by the electrical size of the antenna, typically chosen as $N = \lceil \kappa a \rceil$, where $a$ is the radius of the smallest sphere enclosing the antenna~\cite{Hansen_1988}.
In (\ref{eq:elec_field_as_sum_spherical_func}), $\mathbf{F}^{(3)}_{smn}$ represents the outward-propagation SVWF of mode $(s,m,n)$ and $T_{smn}$ is the corresponding expansion coefficient.

For a specific SVWF $\mathbf{F}^{(c)}_{smn}$, the index $c \in \{ 1,2,3,4\}$, where $c=1$ or $2$ corresponds to standing waves, and $c=3$ and $4$ represent outward- and inward-propagation traveling waves, respectively.
The index $s$ distinguishes the wave polarization type, with $s=1$ denoting transverse electric (TE) modes and $s=2$ denoting transverse magnetic (TM) modes.
The indices $n$ and $m$ characterize the angular dependence of the SVWFs.
Specifically, $n$ determines the polar variation and overall angular complexity, while 
$m$ governs the azimuthal periodicity.
With a set of SRAS spatial parameters $\{r_\mathrm{t}, \theta_\mathrm{t}, \varphi_\mathrm{t}\}$, the explicit expressions of SVWFs are provided in \cite[eq.~(A1.45)]{Hansen_1988} and are therefore omitted here for a neat presentation.

The coefficient $T_{smn}$ represents the antenna's ability to convert an excitation signal $v_i$ into the corresponding SVWF of mode index $(s,m,n)$, also referred to as the radiation coefficient. 
It is fully determined by the antenna's physical properties, e.g., materials, structure, and feeding network, and therefore $T_{smn}$ are typically known a priori for channel modeling.
 
\begin{remark}
It is worth noting that conventional Green's function-based approaches\cite{10500425,10500751} require prior knowledge of the antenna's surface current distribution, which is typically inaccessible in real-world engineering scenarios.
In contrast, SVWE only relies on measurable radiation coefficients $T_{smn}$ to reconstruct the antenna's radiated electric field, making it particularly well-suited for practical channel modeling.
\end{remark}

\subsection{Propagation Characterization: From $\mathbf{E}_\mathrm{t}$ to $\mathbf{E}_\mathrm{r}$}\label{subsec:prop_charc}
In an SRAS-enabled system, variation in the position and orientation of the transmit antenna alters the radiation field $\mathbf{E}_{\mathrm{t}}$, which in turn significantly affects the incident field $\mathbf{E}_{\mathrm{r}}$ at the receive antenna. 
Consequently, the propagation characteristics from $\mathbf{E}_{\mathrm{t}}$ to $\mathbf{E}_{\mathrm{r}}$ must be carefully analyzed to accurately model the impact of transmitter orientation variations on channel gain.
Note that in (\ref{eq:elec_field_as_sum_spherical_func}), we obtain the expression for the radiated electric field $\mathbf{E}_{\mathrm{t}}$ in the coordinate system $(r_\mathrm{t},\theta_\mathrm{t},\varphi_\mathrm{t})$.
To quantitatively analyze the transformation from $\mathbf{E}_{\mathrm{t}}$ to $\mathbf{E}_{\mathrm{r}}$,
we first transform $\mathbf{E}_{\mathrm{t}}$ into the spherical coordinate system $(r_\mathrm{r},\theta_\mathrm{r},\varphi_\mathrm{r})$ associated with $(x_\mathrm{r},y_\mathrm{r},z_\mathrm{r})$, and then isolate the inward-propagating component, which constitutes the incident field $\mathbf{E}_{\mathrm{r}}$ at the receiver.

In fact, owing to the rotation and translation properties of SVWFs, this transformation from $\mathbf{E}_{\mathrm{t}}$ to $\mathbf{E}_{\mathrm{r}}$ can be carried out analytically.
Consider a coordinate system $(x,y,z)$ that undergoes a sequence of rotations: 1) by an angle $\alpha$ about the $z$-axis; 2) by an angle $\beta$ about the rotated $y$-axis; and 3) by an angle $\omega$ about the subsequently rotated $z$-axis, yielding the final rotated system $(x',y',z')$.
The rotation property of SVWF can be expressed as
\begin{equation}\label{eq:rotation_formula}
	\begin{aligned}
  		\mathbf{F}^{(c)}_{smn}(r,\theta,\varphi) = \sum_{\mu=-n}^n D^n_{\mu m}({  \omega,\beta,\alpha})  \mathbf{F}^{(c)}_{s \mu n}(r',\theta',\varphi'),
  	\end{aligned}
\end{equation}
where the rotation coefficient, known as the Wigner D-function, is given by
\begin{equation}\label{eq:D_func}
	\begin{aligned}
  		D^n_{\mu m}({  \omega,\beta,\alpha}) = e^{\jmath m \alpha} d^n_{\mu m}(\beta) e^{\jmath \mu \omega}.
  	\end{aligned}
\end{equation}
In particular, $d^n_{\mu m}(\beta)$ is a complex-valued function, and its explicit expression is provided in \cite[eq.~(A2.5)]{Hansen_1988}.

If the coordinate system $(x,y,z)$ is translated by a distance $A$ along the $z$-axis to obtain $(x',y',z')$, the translation property of SVWF can be expressed as
\begin{equation}\label{eq:translation_formula}
	\begin{aligned}
  		\mathbf{F}_{{s \mu n}}^{(c)}(r, \theta, \varphi) = 
								\sum_{\sigma=1}^{2}
								\sum_{\substack{\nu=|\mu| \\ \nu \neq 0}}^{N}
								C_{\sigma \mu \nu}^{{sn (c)}}(\kappa A)
								\mathbf{F}_{\sigma \mu \nu}^{(1)}(r', \theta', \varphi')
  	\end{aligned}
\end{equation}
for $r' < |A|$.
The translation coefficient $C_{\sigma \mu \nu}^{{sn(c)}}(\kappa A)$ is a complex-valued function, and its explicit expression is provided in \cite[eq.~(A3.3)]{Hansen_1988}.

\textcolor{black}{In essence, the above rotation and translation formulas provide the mathematical machinery for expressing $\mathbf{E}_{\mathrm{t}}$ in the receiver's local coordinates. Specifically, the rotation property accounts for the relative orientation between the two antennas, while the translation property handles their spatial separation. Once $\mathbf{E}_{\mathrm{t}}$ is re-expressed in $(r_\mathrm{r},\theta_\mathrm{r},\varphi_\mathrm{r})$, it takes the form of a superposition of standing waves. This standing-wave form arises from the translation property in (\ref{eq:translation_formula}), which maps spherical waves to a superposition of standing waves ($c=1$) in the translated coordinate system. Each standing wave can be further decomposed into an outward-propagating traveling wave ($c=3$) and an inward-propagating traveling wave ($c=4$). Physically, only the inward-propagating component reaches the receive antenna port and induces a response, and this component is precisely the incident field $\mathbf{E}_{\mathrm{r}}$. 
}

\subsection{Reception Characterization: From $\mathbf{E}_\mathrm{r}$ to $w_j$}
After the propagation process is analyzed in Sec. \ref{subsec:prop_charc}, the incident field $\mathbf{E}_{\mathrm{r}}$ at the receive antenna can be expressed in the receive coordinate system as a superposition of inward-propagating SVWFs. 
Therefore, we have
\begin{equation}\label{eq:Er}
	\begin{aligned}
  		\mathbf{E}_{\mathrm{r}}(r_\mathrm{r},\theta_\mathrm{r},\varphi_\mathrm{r})= \sum_{s=1}^{2}\sum_{n=1}^{N}\sum_{m=-n}^{n}   Q_{smn}\mathbf{F}_{smn}^ {(4)} (r_\mathrm{r},\theta_\mathrm{r},\varphi_\mathrm{r}),
  	\end{aligned}
\end{equation}
where $Q_{smn}$ is the expansion coefficients.  Accordingly, the received signal $w_j$ induced by $\mathbf{E}_{\mathrm{r}}$ at the $j$-th receive antenna can be expressed as
\begin{equation}\label{eq:Er_receive_signal_form}
	\begin{aligned}
  		w_j = \sum_{s=1}^{2}\sum_{n=1}^{N}\sum_{m=-n}^{n} Q_{smn} R_{smn},
  	\end{aligned}
\end{equation}
where $R_{smn}$ is the reception coefficient, characterizing the antenna's ability to convert an incident spherical mode $(s,m,n)$ into the receive signal at ports.

For reciprocal antennas, which constitute the vast majority of practical radiating elements, the reception and radiation characteristics are related by electromagnetic reciprocity.
Specifically, the reception coefficient $R_{smn}$ can be directly determined from the radiation coefficient $T_{smn}$~\cite{Hansen_1988}, namely
\begin{equation}
	\begin{aligned}
  		R_{smn} = (-1)^m T_{s,-m,n}.
  	\end{aligned}
\end{equation}

Building upon the above analysis of the three fundamental physical processes, i.e., radiation, propagation, and reception, that determine the EM characteristics of the wireless channel, we have established a foundational understanding of the EM-based channel model. 
This provides the theoretical foundation for developing a channel model tailored to SRAS-enabled MIMO systems.


\section{EM-Based Channel Model for SRAS-Enabled MIMO}\label{sec:EM-based_channel_model}

Building on the SVWE analysis in Sec. ~\ref{sec:EM_char_analysis}, we establish the EM-based channel model for SRAS-enabled MIMO systems. We first derive a closed-form channel gain under the line-of-sight (LoS) assumption, which clearly reveals the role of antenna type, orientation, and relative position. The model is then extended to general multipath channels.

\subsection{LoS Channel Model}\label{sec:LoS_model}

To enable a compact matrix formulation, we introduce a one-to-one index mapping $(s,m,n) \leftrightarrow  l$ that compresses the multi-dimensional SVWF mode indices into a single index~\cite{Hansen_1988}, and denote the total number of spherical vector wave modes as $L$. By stacking the radiation coefficients $T_{smn}$ and receive coefficients $R_{smn}$ into vectors $\mathbf{t} \in \mathbb{C}^{L \times 1}$ and $\mathbf{r} \in \mathbb{C}^{L \times 1}$, respectively, the following proposition provides an explicit expression for $h_{ji}$.

\begin{prop}\label{prop:hij}
{\color{black}{
	The EM-based LoS channel gain $h_{ji}$ is given by
\begin{equation}\label{eq:EM_based_SISO_channel_model}
	\begin{aligned}
  		h_{ji} = \frac{1}{2\sqrt{\eta}}\frac{e^{\jmath \kappa \Vert \mathbf{d}_{ji} \Vert}}{{ \Vert \mathbf{d}_{ji} \Vert}}
  		\mathbf{r}^T \mathbf{P}_{ji} \mathbf{O}_i \mathbf{t}.
  	\end{aligned}
\end{equation}
By leveraging the index mappings $(s,\mu,n) \leftrightarrow p$ and $(\sigma,\rho,\nu) \leftrightarrow q$,
the $p,q$-th element of $\mathbf{O}_i \in \mathbb{C}^{L \times L} $ is given by
\begin{equation}\label{eq:Go_pq}
	\begin{aligned}
  		[\mathbf{O}_i]_{p,q} = \delta_{s,\sigma}\,\delta_{n,\nu}\, D^{n}_{\mu \rho}(-\omega_i,-\beta_i,-\alpha_i).
  	\end{aligned}
\end{equation}
}}
The ${q},{p}$-th element of the matrix $\mathbf{{P}}_{ji} \in \mathbb{C}^{L \times L}$ is given by
\begin{equation}\label{eq:G_pq}
	\begin{aligned}
  		[\mathbf{{P}}_{ji}]_{{  q},{  p}} = 
  		a(n,\nu) 
  		\frac{\kappa}{2\sqrt{\eta}}
  		e^{\jmath ( \mu - \rho)\varphi_{ji}} 
  		b^{s,\mu,n}_{\sigma,\rho,\nu}(\theta_{ji}),
	\end{aligned}
\end{equation}
where
\begin{equation}\label{eq:anv}
	\begin{aligned}
  		a(n,\nu) = \frac{\jmath ^{\nu-n-1}\sqrt{(2n+1)(2\nu+1)}}{2},
  	\end{aligned}
\end{equation}

\begin{equation}
	\begin{aligned}\label{eq:b_theta}
  		b^{s,\mu,n}_{\sigma,\rho,\nu}(\theta_{ji}) = 
  		d^{n}_{1,\mu}(\theta_{ji}) & d^{\nu}_{1,\rho}(\theta_{ji}) + \\
  		& (-1)^{s+\sigma} d^{n}_{-1,\mu}(\theta_{ji}) d^{\nu}_{-1,\rho}(\theta_{ji})
  	\end{aligned}
\end{equation}

\end{prop}

\begin{IEEEproof}
	Please refer to Appendix \ref{sec:appendix_prof_prop_hij}.
\end{IEEEproof}

\begin{remark}
	The EM-based channel gain in (\ref{eq:EM_based_SISO_channel_model}) is expressed as the product of two distinct factors. The first factor, ${e^{\jmath \kappa \|\mathbf{d}_{ji}\|}}/{\|\mathbf{d}_{ji}\|}$, represents a \emph{scalar} spherical wave that depends solely on the distance between the transmit and receive antennas. 
	The second factor, $\mathbf{r}^T \mathbf{P}_{ji} \mathbf{O}_i \mathbf{t}$, captures the \emph{vectorial} nature of EM propagation and incorporates the specific antenna types (reflected by $\mathbf{t}$ and $\mathbf{r}$), the orientation of the transmit antenna ($\mathbf{O}_i$), and the relative angular position between the transceivers ($\mathbf{P}_{ji}$). 
	Classical channel models~\cite{tseWirelessComm} typically retain only the scalar distance-dependent factor.
	Consequently, they fail to capture the variations in channel gain induced by antenna orientation and polarization.
	In contrast, the proposed EM-based channel model intrinsically accounts for these vectorial EM effects\textcolor{black}{, and moreover, cleanly isolates the antenna orientation and propagation environment into the distinct matrices $\mathbf{O}_i$ and $\mathbf{P}_{ji}$, respectively}, thereby enabling efficient system design in SRAS-enabled MIMO systems in the next section.
\end{remark}

\subsection{Extension to Multipath Channels}\label{sec:multipath_extension}

Before extending the LoS channel model in (\ref{eq:EM_based_SISO_channel_model}) to multipath scenarios, it is helpful to clarify the physical meaning of $\mathbf{P}_{ji}$. Its $({  q},{  p})$-th entry describes the contribution of the ${  p}$-th transmit spherical mode to the ${  q}$-th receive spherical mode after propagating through the channel. With this interpretation, Proposition~\ref{prop:hij} completely decouples the propagation environment (captured in $\mathbf{P}_{ji}$) from the antenna's EM properties and their orientations (handled by $\mathbf{t}$, $\mathbf{r}$, and $\mathbf{O}_i$). Consequently, different propagation scenarios can be modeled simply by plugging in a suitable $\mathbf{P}_{ji}$. We outline two representative cases below.

\subsubsection{Deterministic Multipath with Discrete Scatterers}
\textcolor{black}{We consider the single-bounce case, where each propagation path involves exactly one scattering event.}
Suppose the environment contains $N_{\mathrm{s}}$ well-resolved scatterers. The propagation matrix $\mathbf{P}_{ji}$ can then be constructed as a superposition of the LoS and scattered components:
\begin{equation}
    \mathbf{P}_{ji} = \mathbf{P}_{\mathrm{LoS},ji} + \sum_{n_{\mathrm{s}}=1}^{N_{\mathrm{s}}} \mathbf{P}_{{n_{\mathrm{s}}},ji}.
\end{equation}
Each $\mathbf{P}_{{n_{\mathrm{s}}},ji}$ admits a cascaded structure comprising three EM processes: the LoS propagation from the transmitter to the scatterer, described by the matrix $\mathbf{P}_{\mathrm{ts},l}$; scattering at the scatterer itself; and the LoS propagation from the scatterer to the receiver, described by the matrix $\mathbf{P}_{\mathrm{sr},l}$. Within the SVWE framework, the scattering behavior of a scatterer is fully characterized by its scattering matrix $\mathbf{S}_{l}$, which maps the incident spherical mode coefficients to the outgoing mode coefficients. \textcolor{black}{The entries of $\mathbf{S}_{l}$ are determined by the scatterer's material properties (permittivity, permeability, conductivity), geometry, and physical dimensions.} For scatterers of canonical shape such as spheres and ellipsoids, $\mathbf{S}_{l}$ admits closed forms~\cite{1445988,Gouesbet2010}, whereas for irregular scatterers it can be computed using numerical methods~\cite{1406246,6231655}. Consequently, we have
\begin{equation}
    \mathbf{P}_{{n_{\mathrm{s}}},ji} = \mathbf{P}_{\mathrm{ts},l} \mathbf{S}_{l} \mathbf{P}_{\mathrm{sr},l}.
\end{equation}

\subsubsection{Statistical Multipath Channels}
In stochastic multipath environments, the SVWE coefficients of the incident EM field at the receiver are random variables~\cite{5159555,9091906}. For a Rayleigh fading channel, it has been proven in \cite{5159555} that these coefficients follow a zero-mean circularly symmetric complex Gaussian distribution, whose covariance is determined by the power angular spectrum (PAS) of the propagation environment. The PAS can be obtained from either field measurements or standardized channel models. Consequently, $\mathbf{P}_{ji}$ is a random matrix with jointly Gaussian entries. A Rician fading channel further incorporates a deterministic LoS component, yielding
\begin{equation}
    \mathbf{P}_{ji} = \sqrt{\frac{K_{\mathrm{R}}}{K_{\mathrm{R}}+1}} \mathbf{P}_{\mathrm{LoS},ji} + \sqrt{\frac{1}{K_{\mathrm{R}}+1}} \tilde{\mathbf{P}}_{ji},
\end{equation}
where $K_{\mathrm{R}}$ is the Rician $K$-factor and $\tilde{\mathbf{P}}_{ji}$ is a zero-mean complex Gaussian random matrix.

The above cases illustrate that the proposed SVWE-based channel framework, by virtue of decoupling antenna properties from the propagation environment, can accommodate diverse channel models without altering the core formulation in (\ref{eq:EM_based_SISO_channel_model}). This generality is a key advantage of the SVWE-based approach over existing EM-based channel models.

\subsection{Impact of Antenna Movement versus Orientation on Achievable Rate}\label{subsec:impact_on_commRate}

The closed-form channel gain in Proposition~\ref{prop:hij} reveals that antenna position and orientation influence $h_{ji}$ through distinct mechanisms: position determines the scalar distance factor and the propagation matrix $\mathbf{P}_{ji}$, while orientation determines the orientation matrix $\mathbf{O}_i$. 
To assess their relative impact on communication performance, we compare the achievable rate variation induced by antenna displacement with that induced by antenna rotation through the following two experiments.
1) the antenna's orientation is fixed at $\alpha = \ang{0}$, $\beta = \ang{0}$, while its position is varied in the $x\mathcal{O}y$-plane over a large interval $[-20\lambda,20\lambda] \times [-20\lambda,20\lambda]$;
2) the transmit antenna's position is fixed at $\mathcal{O}$ while the angles $\alpha$ and $\beta$ are varied over $[0,\pi]$.
\textcolor{black}{As a representative example, both the transmit and receive antennas are modeled as half-wavelength dipoles.}
\textcolor{black}{The spin angle $\omega$ is fixed at $\ang{0}$, as it does not affect the radiation of the half-wavelength dipole antennas.}
The communication performance is evaluated by the achievable rate, defined by $\log_2 \left( 1 + P_\mathrm{t} \vert h_{ji} \vert^2\ / \sigma_{\mathrm{n}}^2 \right)$,
where $P_\mathrm{t} = 1$ W is the transmit power and $\sigma_{\mathrm{n}}^2 = -40$ dBm denotes the noise power.

The simulation results are shown in Fig.~\ref{fig:rate_comparison_pos_vs_ori}.
It can be observed that even when the transmit antenna is moved over a considerable range, the achievable rate exhibits only slight variation.
In contrast, adjusting the orientation induces significant fluctuations in the achievable rate.
Specifically, in the considered simulation scenario, the peak-to-peak variation of the achievable rate caused by antenna rotation is
$3.85$~bps/Hz, whereas that induced by antenna movement is only $0.81$~bps/Hz.
This result reveals that, compared to antenna movement, rotation exerts a far more pronounced influence on the achievable rate.

\begin{figure}[t]
	\centering
	\vspace{0pt}
	\hspace{-3pt}\includegraphics[width=0.5\textwidth]{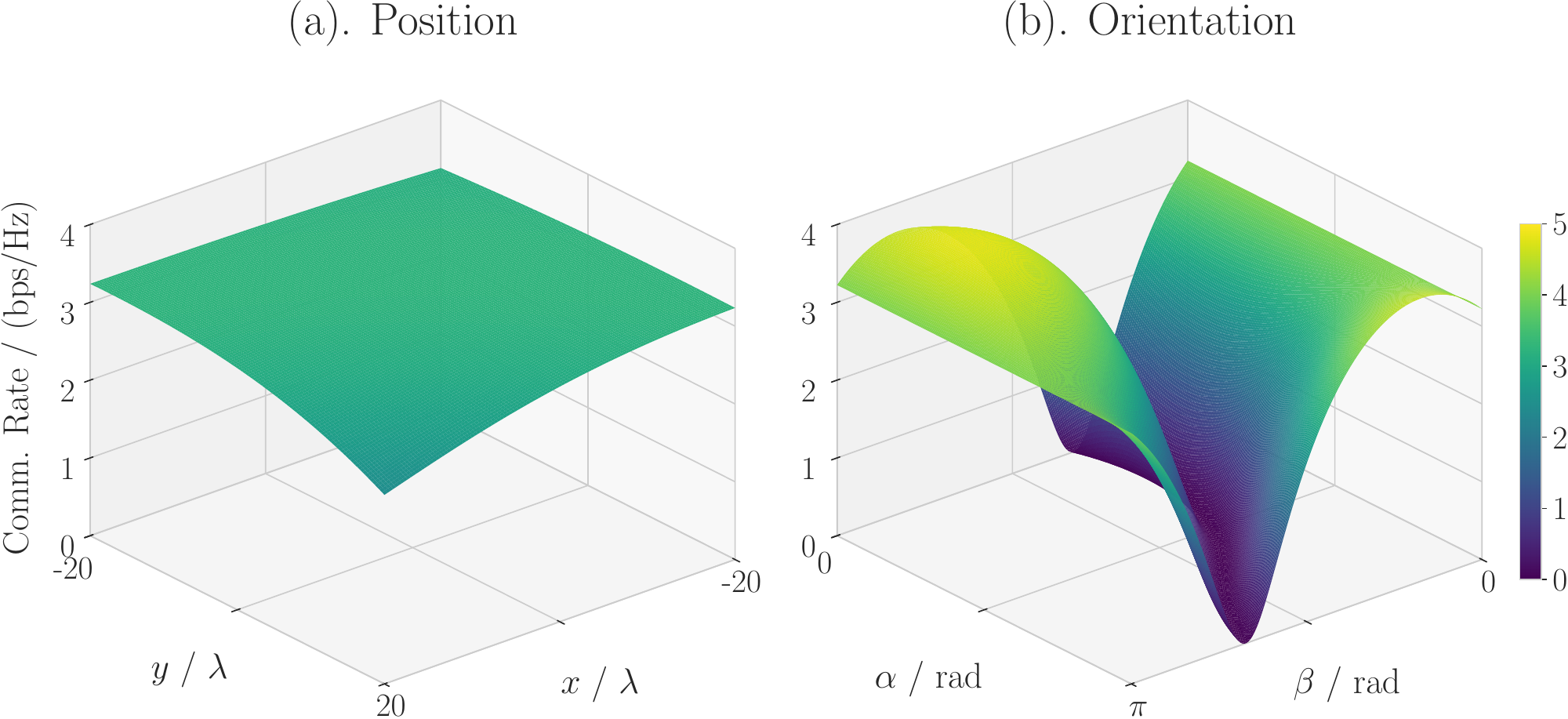}
	\caption{The variation of the point-to-point channel achievable rate for the transmit antenna with different orientations or positions.}
	\label{fig:rate_comparison_pos_vs_ori}
\end{figure}

This result can be understood from Proposition~\ref{prop:hij}. 
\textcolor{black}{Although the matrices $\mathbf{O}_i$ and $\mathbf{P}_{ji}$ appear symmetrically in the vectorial factor $\mathbf{r}^T \mathbf{P}_{ji} \mathbf{O}_i \mathbf{t}$, they differ fundamentally in the angular ranges they can access. The orientation matrix $\mathbf{O}_i$ is directly parameterized by $\alpha_i$, $\beta_i$, and $\omega_i$, \textcolor{black}{which cover the ranges $\alpha_i,\omega_i\in[0,2\pi]$ and $\beta_i\in[0,\pi]$, }thereby covering the full polarization range and driving the inner product from its maximum to near zero. In contrast, the propagation matrix $\mathbf{P}_{ji}$ depends on the geometric angles $\theta_{ji}$ and $\varphi_{ji}$, which are determined by the relative transceiver position. Under practical deployment conditions, antenna displacement changes these geometric angles only marginally, so the accessible dynamic range of $\mathbf{P}_{ji}$ is far narrower than that of $\mathbf{O}_i$. Consequently, despite their symmetric role in the inner product, orientation exerts a substantially stronger influence because $\mathbf{O}_i$ admits much wider angular coverage than $\mathbf{P}_{ji}$.}

\section{Sum-Rate Maximization in SRAS-Enabled MU-MIMO Systems}\label{sec:sum_rate_opt}

As shown in Sec. \ref{subsec:impact_on_commRate}, antenna orientation exerts a stronger influence on the achievable rate than antenna displacement. 
Motivated by this finding, this section investigates sum-rate maximization in an SRAS-enabled MU-MIMO system by jointly optimizing the orientation of all transmit antennas.
We formulate the problem directly on the special orthogonal group $\mathrm{SO}(3)^{N_{\mathrm{t}}}$, which is the natural domain of the orientation variables~\cite{Hall2003}, and develop a Riemannian gradient ascent method ~\cite{absil2009optimization} that respects this geometric structure to obtain a locally optimal solution to this highly non-convex problem. 

\subsection{Problem Formulation}

We consider the downlink of an SRAS-enabled $K$-user MIMO system, where the base station (BS) is equipped with $N_{\mathrm{t}}$ reconfigurable antenna elements and each user is equipped
  with $N_{\mathrm{u}}$ fixed antennas.
  \textcolor{black}{For the $k$-th user, let $\mathbf{d}_{ji}^{(k)} \in \mathbb{R}^{3 \times 1}$ denote the displacement vector from the $i$-th transmit antenna to its $j$-th receive antenna, for $i = 1,\dots,N_{\mathrm{t}}$, $j = 1,\dots,N_{\mathrm{u}}$, and $k=1,\dots, K$.
  }
  These relative positions are assumed known a priori.
  The received baseband signal $\mathbf{w}_k \in \mathbb{C}^{N_{\mathrm{u}} \times 1}$ at the $k$-th user is
  \begin{equation}
      \mathbf{w}_k  = \mathbf{H}_k \mathbf{F}_k \mathbf{v}_k + \sum_{u \neq k} \mathbf{H}_k \mathbf{F}_u \mathbf{v}_u + \mathbf{z},
  \end{equation}
  where $\mathbf{v}_k \in \mathbb{C}^{N_{\mathrm{u}} \times 1}$ is the transmitted symbol vector, $\mathbf{H}_k \in \mathbb{C}^{N_{\mathrm{u}} \times N_{\mathrm{t}}}$ is the EM-based channel matrix from the BS to the $k$-th user, which is constructed from Proposition~\ref{prop:hij} given $\{\mathbf{d}_{ji}^{(k)}\}$ and the orientations of the transmit antennas, and $\mathbf{F}_k \in \mathbb{C}^{N_{\mathrm{t}} \times N_{\mathrm{u}}}$ is
  the precoding matrix of the $k$-th user.

  To mitigate multiuser interference, we adopt the MMSE precoding scheme.
  Let $\mathbf{H} = [\mathbf{H}_1^T, \mathbf{H}_2^T, \dots, \mathbf{H}_K^T]^T \in \mathbb{C}^{K N_{\mathrm{u}} \times N_{\mathrm{t}}}$ be the overall channel matrix.
  The MMSE precoder is given by
  \begin{equation}\label{eq:W_MMSE}
      \mathbf{F}_{\mathrm{MMSE}} = \zeta \mathbf{H}^H (\mathbf{H}\mathbf{H}^H + \epsilon \mathbf{I}_{K N_{\mathrm{u}}})^{-1} \triangleq [\mathbf{F}_1, \mathbf{F}_2, \dots, \mathbf{F}_K],
  \end{equation}
  where $\epsilon = K \sigma_\mathrm{n}^2 / P$ is the regularization coefficient, $P$ is the power budget, and $\zeta = {\sqrt{P}}/{\Vert \mathbf{H}^H (\mathbf{H}\mathbf{H}^H +
  \epsilon \mathbf{I}_{K N_{\mathrm{u}}})^{-1} \Vert_F}$ enforces the power constraint.
  Defining $\mathbf{Q}_k = \mathbf{F}_k \mathbf{F}_k^H$ and $\mathbf{N}_k = (\sigma_{\mathrm{n}}^2 \mathbf{I}_{N_\mathrm{u}} + \sum_{j \neq k} \mathbf{H}_k \mathbf{Q}_j
  \mathbf{H}_k^H)^{-1}$, the achievable rate of the $k$-th user is given by~\cite{1261332}
  \begin{equation}\label{eq:R_k}
      R_k = \log_2 \det \bigl( \mathbf{I}_{N_\mathrm{u}} + \mathbf{H}_k \mathbf{Q}_k \mathbf{H}_k^H \mathbf{N}_k \bigr).
  \end{equation}

  Our objective is to maximize the sum-rate by optimizing the orientation of each transmit antenna.
   In the previous section, we used the angle tuple $(\alpha_i,\beta_i,\omega_i)$ to parameterize the orientation of the $i$-th transmit antenna.
    {However, these three angles are coupled, namely, a change in one alters the rotation axis of a subsequent angle. This coupling is an artifact of the angle-based parameterization rather than a physical constraint. 
    We can therefore bypass the coupling entirely by treating an orientation not as a triple of angles, but as a single element of the special orthogonal group $\mathrm{SO}(3)$, i.e., the set of all $3\times3$ matrices satisfying $\mathbf{U}\mathbf{U}^T = \mathbf{I}_3$ and $\det(\mathbf{U}) = 1$. 
    {\color{black}{Given a matrix $\mathbf{U} \in \mathrm{SO}(3)$, the corresponding rotation angles can be recovered as
    \begin{equation}\label{eq:U_2_algles}
        \begin{aligned}
            \alpha &= \operatorname{atan2}\bigl([\mathbf{U}]_{2,3},\,[\mathbf{U}]_{1,3}\bigr), \\
            \beta  &= \arccos\bigl([\mathbf{U}]_{3,3}\bigr), \\
            \omega &= \operatorname{atan2}\bigl([\mathbf{U}]_{3,2},\,-[\mathbf{U}]_{3,1}\bigr),
        \end{aligned}
    \end{equation}
    where $\operatorname{atan2}(\cdot,\cdot)$ denotes the two-argument arctangent.}}
    As a smooth manifold, $\mathrm{SO}(3)$ admits a well-defined gradient direction for optimizing antenna orientations. By formulating the optimization directly on the product manifold $\mathrm{SO}(3)^{N_{\mathrm{t}}}$, the three rotational degrees of freedom can be updated simultaneously, free from the artificial coupling of angle-based parameterizations. 
    Let $\mathbf{U}_i \in \mathrm{SO}(3)$ denote the rotation matrix representing the orientation of the $i$-th transmit antenna.
  We therefore formulate the sum-rate maximization problem as
  \begin{equation}\label{eq:opt_prob_manifold}
      \max_{\mathbf{U}_1,\dots,\mathbf{U}_{N_{\mathrm{t}}} \in \mathrm{SO}(3)} \; R = \sum_{k=1}^{K} R_k(\mathbf{U}_1,\dots,\mathbf{U}_{N_{\mathrm{t}}}).
  \end{equation}

\begin{remark}
	Problem (\ref{eq:opt_prob_manifold}) presents two main challenges.
	First, each rotation operation $\mathbf{U}_i$ enters the channel through the matrix $\mathbf{O}_i$, whose entries are the Wigner D-functions $D^n_{\mu m}$ in (\ref{eq:D_func})  evaluated at the rotation angles.
	These functions are highly nonlinear and oscillatory in the orientation angles, making the sum-rate a non-convex function of the orientations and precluding closed-form solutions.
	Second, due to multiuser interference, $R_k$ depends not only on $\mathbf{F}_k$ but also on all interfering precoders $\{\mathbf{F}_u\}_{u\neq k}$, coupling all users. \textcolor{black}{Moreover, each $\mathbf{F}_k$ is determined by the channel $\mathbf{H}$ through the MMSE solution, and $\mathbf{H}$ itself depends nonlinearly on the orientation matrices $\mathbf{O}_i$ (and hence on $\mathbf{U}_i$), further tightening the coupling among users and antennas.}
	These two features make the problem highly non-convex and tightly coupled.
\end{remark}

\subsection{Riemannian Gradient Ascent on $\mathrm{SO}(3)$}

To solve problem (\ref{eq:opt_prob_manifold}), we employ Riemannian gradient ascent.
The key idea is that $\mathrm{SO}(3)$ is a \emph{smooth manifold}, which means that although it is globally curved, each element $\mathbf{U}$ lies in a neighborhood that resembles a flat three-dimensional space, called the \emph{tangent space} $T_{\mathbf{U}}\mathrm{SO}(3)$~\cite{Hall2003}.
A gradient direction can therefore be defined in this tangent space, and the optimization proceeds by taking small steps along the gradient and projecting back onto the manifold~\cite{absil2009optimization}.
\textcolor{black}{It is worth emphasizing that the group element $\mathbf{U}_i$ encapsulates all three angles $(\alpha_i,\beta_i,\omega_i)$ in a single geometric object. The Riemannian gradient of the sum-rate with respect to $\mathbf{U}_i$ therefore captures the steepest-ascent direction jointly over all three angles.}

\textcolor{black}{The sum-rate $R$ is a scalar function defined on the product manifold $\mathrm{SO}(3)^{N_{\mathrm{t}}}$. Its Riemannian gradient with respect to $\mathbf{U}_i$ is therefore a tangent vector in the tangent space at $\mathbf{U}_i$, denoted by $T_{\mathbf{U}_i}\mathrm{SO}(3)$.
A favorable property of $\mathrm{SO}(3)$ is that the Riemannian gradient of the sum-rate $R$ at $\mathbf{U}_i$ can be written as
\begin{equation}\label{eq:grad_Ri}
	\begin{aligned}
		\operatorname{grad}_{\mathbf{U}_i} R = \mathbf{U}_i  {\mathbf{G}}_i,
	\end{aligned}
\end{equation}
where $\mathbf{G}_i$ is a $3\times 3$ skew-symmetric matrix ${\mathbf{G}}_i$ encodes all the directional information of the gradient. The following proposition provides the explicit expression of ${\mathbf{G}}_i$ in terms of the system matrices $\mathbf{H}_k$, $\mathbf{Q}_k$, and $\mathbf{N}_k$ that constitute $R$.
}

\begin{prop}\label{prop:gradient_expansion}
\color{black}{
	The skew-symmetric matrix ${\mathbf{G}}_i$ in (\ref{eq:grad_Ri}) can be expanded as
	\begin{equation}\label{eq:hat_gi}
		\begin{aligned}
			{\mathbf{G}}_i = g_{i,x}\mathbf{G}_x + g_{i,y}\mathbf{G}_y + g_{i,z}\mathbf{G}_z,
		\end{aligned}
	\end{equation}
	where $\mathbf{G}_x$, $\mathbf{G}_y$, and $\mathbf{G}_z$ form the standard basis of $3\times 3$ skew-symmetric matrices. Each coefficient $g_{i,a}$ ($a\in\{x,y,z\}$) gives the rate of change of the sum-rate under an infinitesimal rotation of the $i$-th antenna about its body $a$-axis, and is given by
	\begin{equation}\label{eq:g_ia_expression}
		\begin{aligned}
			g_{i,a} = \sum_{k=1}^{K} \mathrm{Tr}\left( \mathbf{B}_k \bigl( \tilde{\mathbf{C}}_{kk}^{(a,i)} - \mathbf{A}_k \sum_{j\neq k} \tilde{\mathbf{C}}_{kj}^{(a,i)} \bigr) \mathbf{N}_k \right),
		\end{aligned}
	\end{equation}
	where $\mathbf{A}_k = \mathbf{H}_k \mathbf{Q}_k \mathbf{H}_k^H \mathbf{N}_k$, $\mathbf{B}_k = (\mathbf{I} + \mathbf{A}_k)^{-1}$, and $\tilde{\mathbf{C}}_{kj}^{(a,i)}$ involves the perturbation of the channel matrix induced by the orientation derivative of $\mathbf{O}_i$ about the $a$-axis, as derived in Appendix~\ref{sec:appendix_riemannian_gradient}.}
\end{prop}

\begin{IEEEproof}
	Please refer to Appendix~\ref{sec:appendix_riemannian_gradient}.
\end{IEEEproof}

Having obtained the gradient direction, we must take a finite step along it while staying on the manifold.
In Euclidean space, one can simply move along a straight line, but on a curved manifold a straight displacement in the gradient direction would leave the manifold, and the resulting orientation would no longer be a valid rotation.
To stay on the manifold, we need a map that takes a tangent vector and returns a point on the manifold following the natural curved path in that direction.
Such a map is called the \emph{exponential map}~\cite{Hall2003}.
While the exponential map is generally difficult to compute in closed form, for $\mathrm{SO}(3)$ it admits the classical Rodrigues' formula
\begin{equation}
	\begin{aligned}
  		\exp({{\mathbf{G}}_i}) = \mathbf{I}_3 + \frac{\sin\theta}{\theta}\,{{\mathbf{G}}_i} + \frac{1-\cos\theta}{\theta^2}\,{\mathbf{G}}_i^2,
  	\end{aligned}
\end{equation}
where $\theta = \sqrt{g_{i,x}^2+g_{i,y}^2+g_{i,z}^2}$.
With Rodrigues' formula in hand, the update on the manifold takes the form
\begin{equation}\label{eq:update_rule}
	\begin{aligned}
  		\mathbf{U}^{(t+1)} = \mathbf{U}^{(t)} \cdot \exp\bigl(\eta{\mathbf{G}}_i^{(t)}\bigr),
  	\end{aligned}
\end{equation}
where $\eta>0$ is a fixed step size, ${\mathbf{G}}_i^{(t)}$ denotes the gradient direction at iteration $t$.
Algorithm~\ref{alg:Riemannian} summarizes the overall procedure.

\begin{algorithm}[t]
	\caption{\color{black}{Riemannian Gradient Ascent for Problem (\ref{eq:opt_prob_manifold}})}\label{alg:Riemannian}
	\begin{algorithmic}[1]		
		\REQUIRE $N_{\mathrm{t}}$, $N_{\mathrm{u}}$, $K$, $P$, $\sigma_\mathrm{n}^2$, $\{\mathbf{d}_{ji}^{(k)}\}$, $\eta$, $\varepsilon$, $t_{\mathrm{max}}$
		\ENSURE A locally optimal solution $\{\mathbf{U}_i\}_{i=1}^{N_{\mathrm{t}}}$ and the corresponding Euler angles $\{(\alpha_i,\beta_i,\omega_i)\}_{i=1}^{N_{\mathrm{t}}}$.
		
		\STATE Initialise $\mathbf{U}_i^{(0)} \gets \mathbf{I}_3$ for $i=1,\dots,N_{\mathrm{t}}$.
		
		\STATE Compute $\mathbf{H}$, $\mathbf{F}_{\mathrm{MMSE}}$, $R^{(0)}$ at $\{\mathbf{U}_i^{(0)}\}$.
		
		\WHILE {$\Delta_{R} \ge \varepsilon$ and $t \le t_{\mathrm{max}}$} 
		
		\STATE Compute ${\mathbf{G}}_i^{(t)} = \sum_{a\in\{x,y,z\}} g_{i,a}^{(t)} \mathbf{G}_a$ via (\ref{eq:g_ia_expression}).
		
		\STATE Update $\mathbf{U}_i^{(t+1)} \gets \mathbf{U}_i^{(t)} \cdot \exp(\eta{\mathbf{G}}_i^{(t)})$ for all $i$.
		\STATE Calculate $R^{(t+1)}$.
		\STATE Update $\Delta_{R} \gets \vert R^{(t+1)} - R^{(t)} \vert$.
		\STATE $t \gets  t + 1$.
		
		\ENDWHILE
		
		\STATE Recover angle tuple $(\alpha_i,\beta_i,\omega_i)$ from $\mathbf{U}_i$ via (\ref{eq:U_2_algles}).
	\end{algorithmic}
\end{algorithm}

\section{Simulation and Measurement Results}\label{sec:simulation_experiment}

In this section, we validate the accuracy of the proposed EM-based channel model and evaluate the sum-rate performance of the proposed orientation design method via numerical simulations and experimental measurements.

\subsection{Validation of the EM-based Channel Model via Simulations}\label{subsec:numerical_simulation}

\begin{figure}[ht]
	\centering
	\vspace{5pt}
	\hspace{-10pt}\includegraphics[trim=0 0 0 0 clip, width=0.46\textwidth]{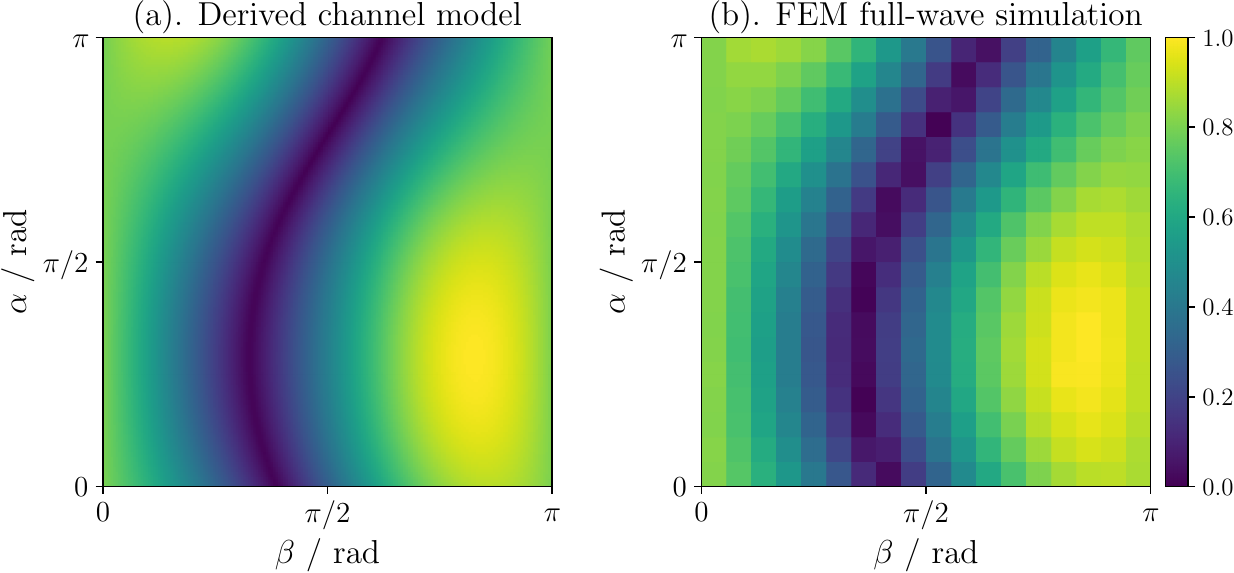}
	\vspace{0mm}
	\caption{The variation of the normalized channel gain with respect to the azimuth angle $\alpha$ and elevation angle $\beta$ of the transmit antenna.}
	\label{fig:channel_gain_SISO}
	
	\vspace{4pt}
	
	\centering
	\vspace{5pt}
	\hspace{-10pt}\includegraphics[trim=0 0 0 0 clip, width=0.46\textwidth]{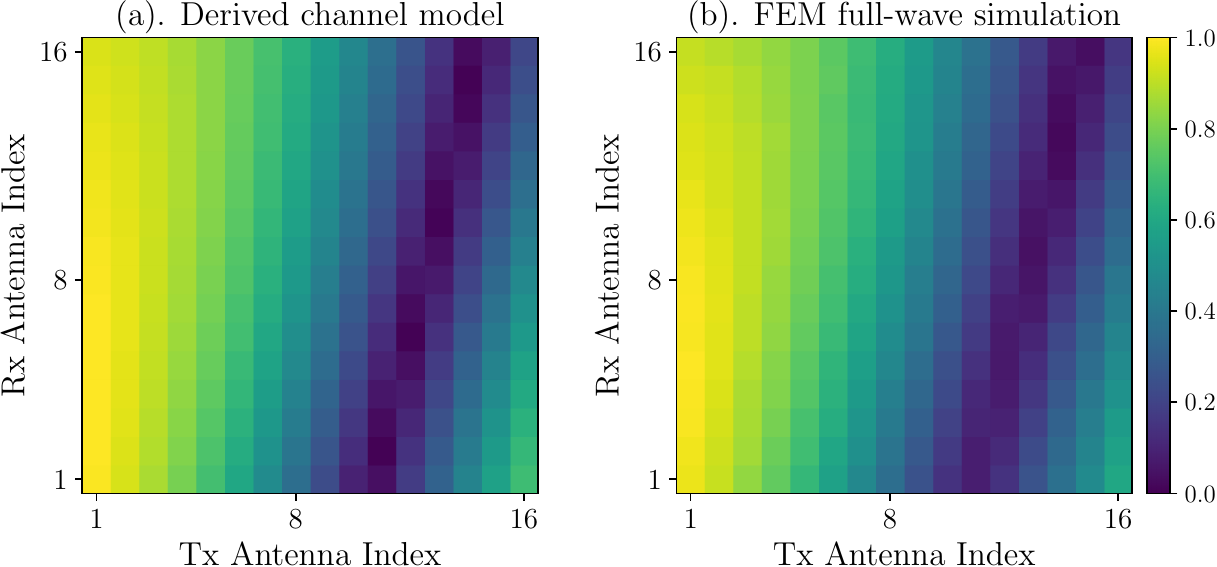}
	\vspace{0mm}
	\caption{The normalized channel gain of the considered SRAS-enabled MIMO system.}
	\label{fig:channel_gain_MIMO}
\end{figure}

\begin{figure*}[t]
\centering
\vspace{2pt}
\begin{minipage}[t]{0.45\linewidth}
    \vspace{0pt}
    \centering
    \includegraphics[width=\linewidth]{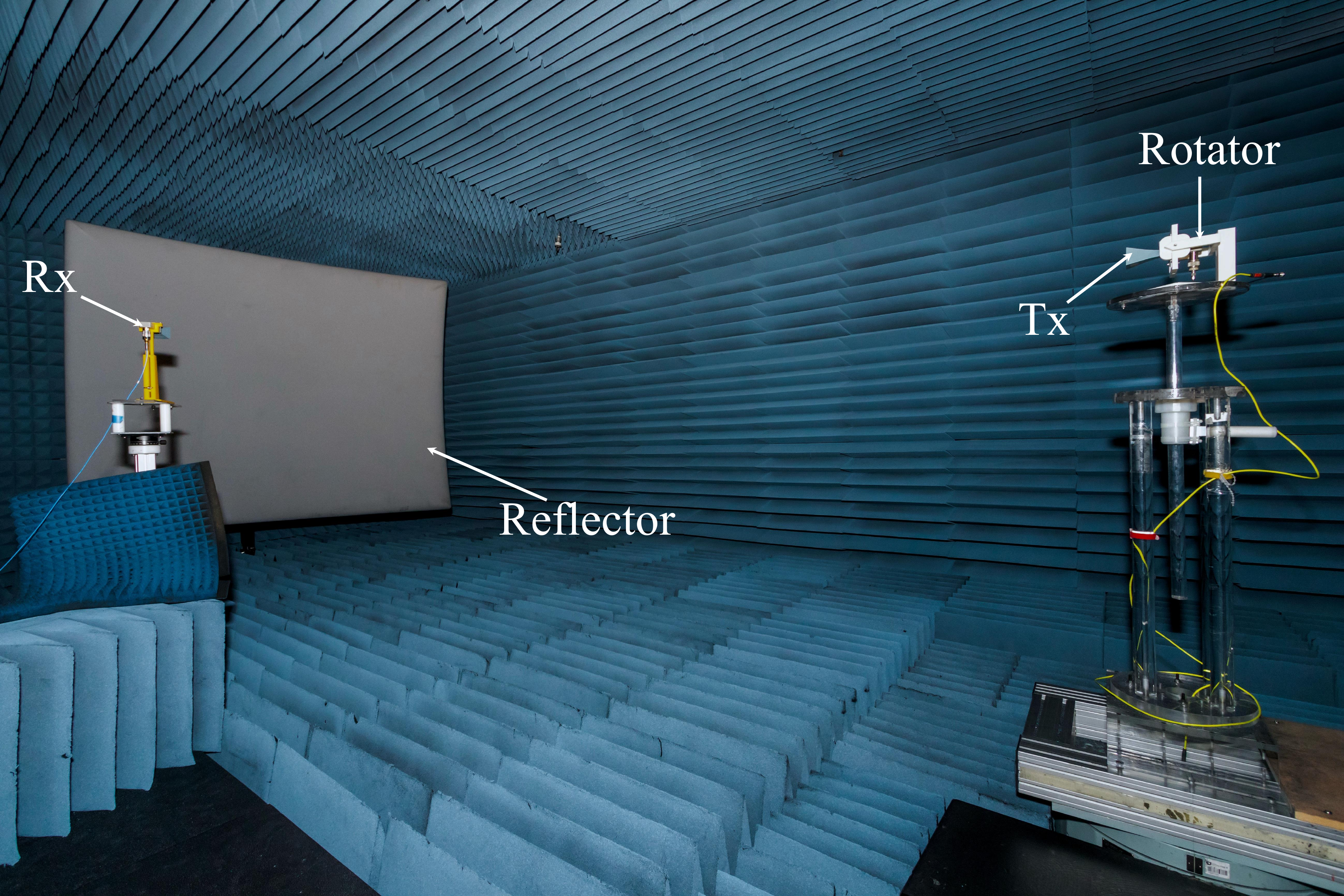}
    \captionsetup{justification=centering}
    \subcaption{}
    \label{fig:exp_all}
\end{minipage}%
\hspace{12pt}
\begin{minipage}[t]{0.32\linewidth}
    \vspace{0pt}
    \centering
    \begin{subfigure}[b]{0.44\linewidth}
        \captionsetup{justification=centering}
        \includegraphics[width=\linewidth]{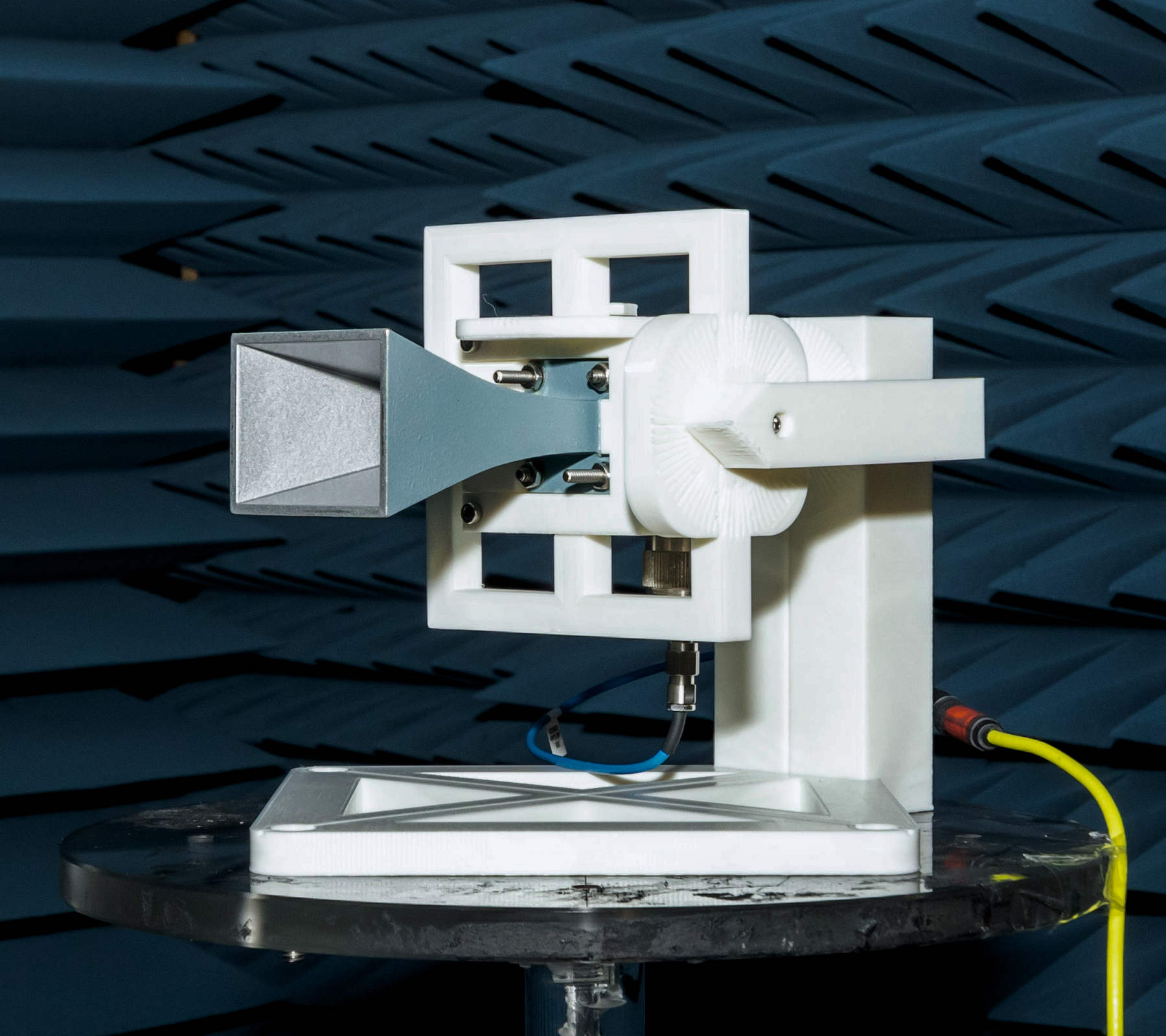}
        \subcaption{}
        \label{fig:Tx_E}
    \end{subfigure}%
    \hfill
    \begin{subfigure}[b]{0.44\linewidth}
        \captionsetup{justification=centering}
        \includegraphics[width=\linewidth]{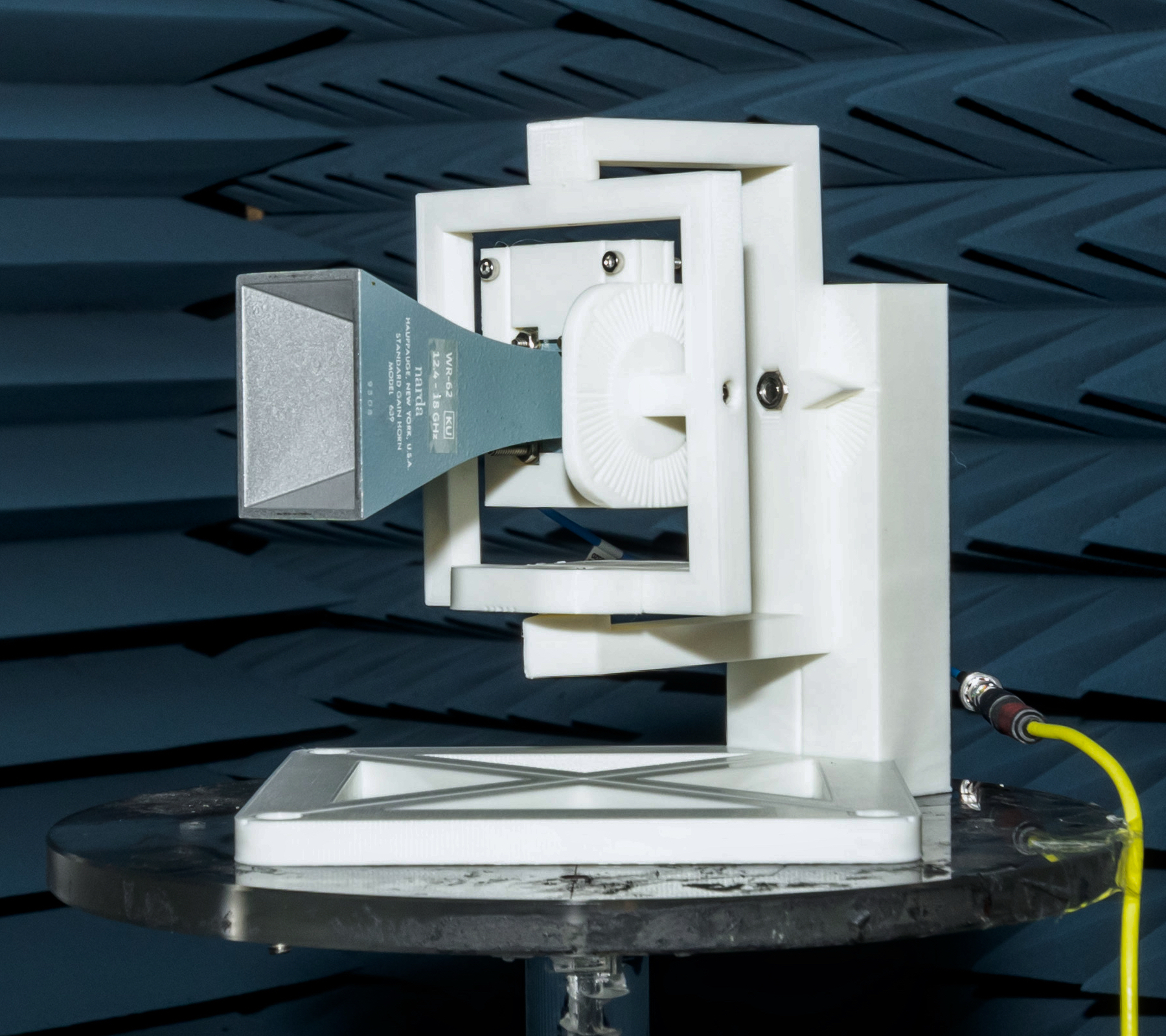}
        \subcaption{}
        \label{fig:Tx_H}
    \end{subfigure}

    \vspace{-0.3pt}

    \begin{subfigure}[b]{\linewidth}
        \centering
        \includegraphics[width=\linewidth]{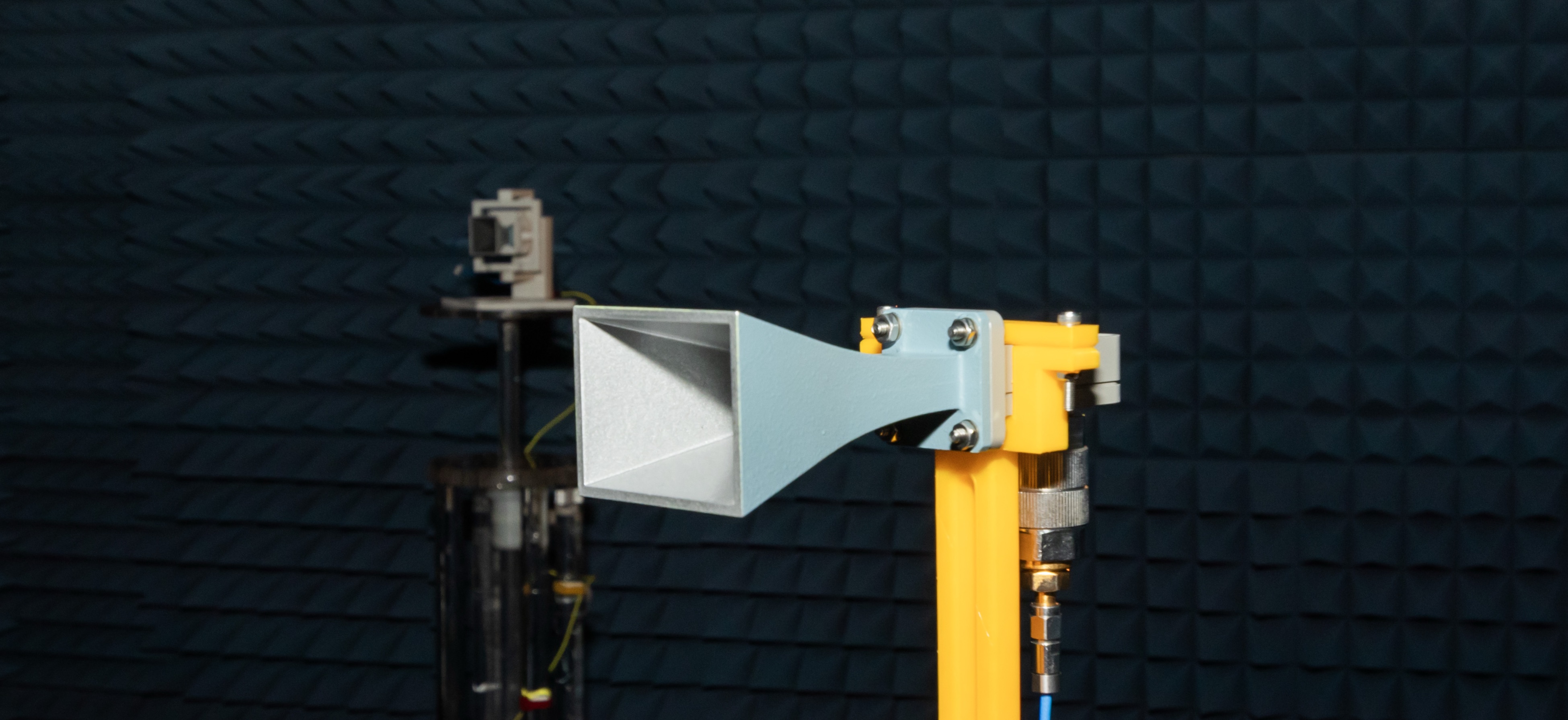}
        \captionsetup{justification=centering}
        \subcaption{}
        \label{fig:Rx}
    \end{subfigure}
\end{minipage}

\caption{(a) The overview of the measurement setup; (b) the transmit antenna orientation corresponding to $\omega = \ang{0}$; (c) the transmit antenna orientation corresponding to $\omega = \ang{90}$; (d) the fixed receive antenna.}
\label{fig:complex_layout}

\end{figure*}

We first validate the capability of the proposed EM-based channel model in capturing the impact of antenna orientation on the channel gain via numerical simulations.
As a representative example, we consider a half-wavelength dipole antenna for both the transmit and receive antennas.
Due to the inherent symmetry of the dipole, its spin angle $\omega$ does not affect the radiation characteristic, thus we focus on the azimuth angle $\alpha$ and elevation angle $\beta$.
The operating wavelength is $\lambda=0.1$~m, the transmit antenna is placed at the origin $\mathcal{O}$, and the receive antenna is located at $\mathcal{P}=(50\lambda,40\lambda,50\lambda)$.

\begin{figure}[t]
	\vspace{0pt}
	\centering
	\hspace{-10pt}\includegraphics[width=0.37\textwidth]{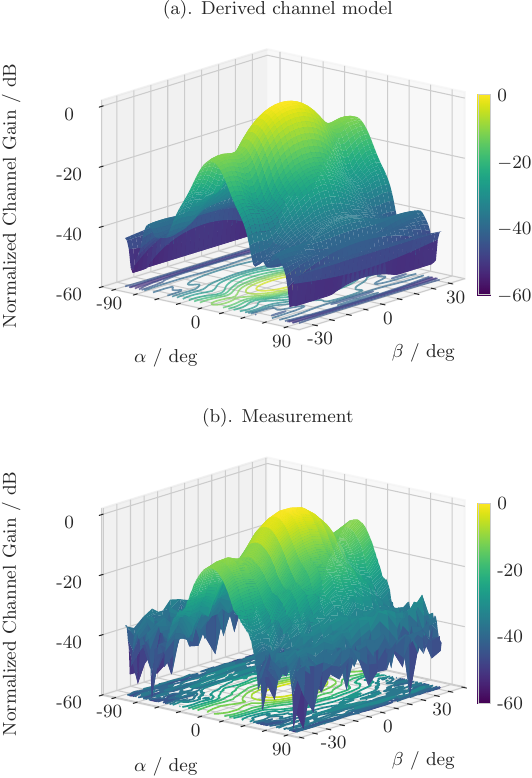}
	\vspace{5pt}
	\caption{The variation of the normalized channel gain with respect to $\alpha$ and $\beta$ of the transmit antenna when $\omega=\ang{0}$.}
	\label{fig:Eplane_sim_vs_mea}
\end{figure}

\begin{figure}[t]
	\vspace{0pt}
	\centering
	\hspace{-10pt}\includegraphics[width=0.37\textwidth]{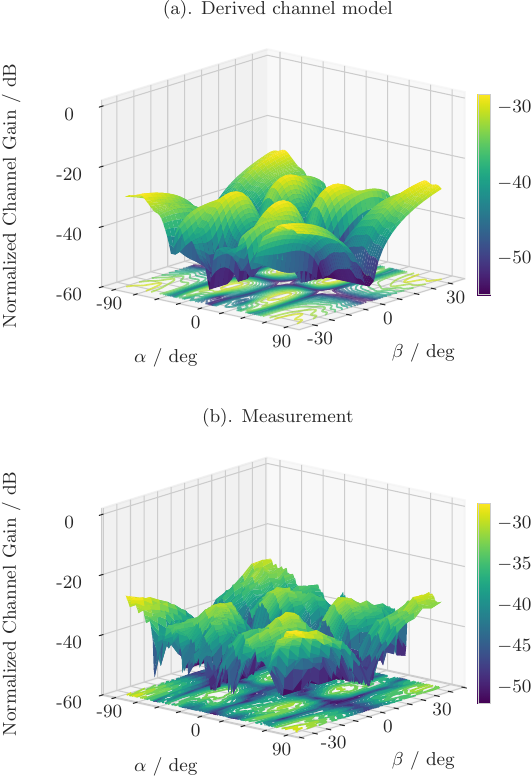}
	\vspace{5pt}
	\caption{The variation of the normalized channel gain with respect to $\alpha$ and $\beta$ of the transmit antenna when $\omega=\ang{90}$.}
	\label{fig:Hplane_sim_vs_mea}
\end{figure}

Fig.~\ref{fig:channel_gain_SISO} shows the point-to-point channel gain computed from (\ref{eq:EM_based_SISO_channel_model}) for different values of $\alpha$ and $\beta$.
For comparison, we employ the finite element method (FEM), a widely used full-wave EM simulation technique. \textcolor{black}{We implement the FEM simulations in Ansys HFSS, a commercial EM solver that inherently captures vectorial EM effects, including polarization mismatch, and thus serves as an independent validation reference.} The corresponding channel gain is computed at every $\ang{10}$ increment in both $\alpha$ and $\beta$.
The normalized mean square error (NMSE) between the two results is approximately $-30.75$~dB.

We then extend the validation to a $16 \times 16$ SRAS-enabled MIMO system, where both the transmitter and receiver are equipped with half-wavelength-spaced uniform linear arrays (ULAs).
Both arrays are aligned along the positive $x$-axis, with their first elements located at $\mathcal{O}$ and $\mathcal{P}$, respectively.
To highlight the reconfigurability of the transmit array, we set the orientation of the $i$-th transmit antenna to $\beta_i = 5i^{\circ}$ and $\alpha_i = \ang{90}$.
Fig.~\ref{fig:channel_gain_MIMO} shows the resulting EM-based channel matrix.
The corresponding FEM result was obtained by computing the channel gain between every transmit-receive antenna pair.
The NMSE between the two results is approximately $-24.13$~dB.

Both Fig.~\ref{fig:channel_gain_SISO} and Fig.~\ref{fig:channel_gain_MIMO} show excellent agreement between the analytically derived channel gain and the FEM-based simulation results, thereby validating the accuracy of the developed EM-based channel model for SRAS-enabled MIMO systems.


\subsection{Validation of the EM-based Channel Model via Measurements}\label{subsec:measurement}

To further evaluate the accuracy of the proposed EM-based channel model, in this section, we present measurement results of channel gain between a rotatable transmit antenna and a fixed receive antenna.
Fig. \ref{fig:exp_all} provides an overview of the measurement setup in a microwave compact antenna test range anechoic chamber.
\textcolor{black}{Both transmit and receive antennas are Ku-band linearly polarized horn antennas (Narda WR-62) with an aperture dimension of {57 }{mm} $\times$ {43 }{mm}, operated at $18$~GHz.}
The transmit antenna is mounted on a 3D-printed rotator to enable independent control of its azimuth angle $\alpha$ and elevation angle $\beta$.
Given the highly directional radiation pattern of the horn antenna, we restrict the measurement range to $\alpha \in [\ang{-90},\ang{90}]$ and $\beta \in [\ang{-30},\ang{30}]$. \textcolor{black}{The chosen ranges span several times the half-power beamwidth of the horn antenna, which is approximately \ang{23}.} Outside this angular region, the channel gain is negligible and dominated by noise.
To further investigate the impact of spin angle $\omega$ on channel gain, the transmit antenna is oriented as illustrated in Fig. \ref{fig:Tx_E} and \ref{fig:Tx_H}, corresponding to $\omega=\ang{0}$ and $\omega=\ang{90}$, respectively.
The receive antenna, as shown in Fig. \ref{fig:Rx}, is fixed at the same height as the transmit antenna.
\textcolor{black}{A parabolic reflector placed {4.8 }{m} from the transmitter and {1.2 }{m} from the receiver converts the spherical wave from the transmit antenna into a plane wave incident on the receive antenna, thereby creating far-field propagation conditions (the far-field distance $2D^2/\lambda \approx {0.6}${ m} is significantly smaller than both path segments). Since the reflector preserves polarization, any measured polarization mismatch originates solely from the relative orientation between the antennas.}

Figs.~\ref{fig:Eplane_sim_vs_mea} and \ref{fig:Hplane_sim_vs_mea} compare the channel gain predicted by the EM-based model to the corresponding measurement results.
First, it can be observed that the theoretical and measured channel gains show good agreement across different transmit antenna orientations.
Furthermore, by comparing the results in Figs.~\ref{fig:Eplane_sim_vs_mea} and \ref{fig:Hplane_sim_vs_mea}, we observe that the channel gain decreases significantly when the transmit antenna's spin angle $\omega$ is changed from $\ang{0}$ to $\ang{90}$.
This phenomenon stems from the difference in polarization alignment between the transmit and receive antennas.
When $\omega = \ang{0}$, the antennas are polarization-matched.
In this case, variations in channel gain induced by adjusting $\alpha$ and $\beta$ are governed mainly by the radiation pattern of the horn antenna.
Hence, the channel gain observed in Fig.~\ref{fig:Eplane_sim_vs_mea} exhibits distinct directional characteristics.
In contrast, when $\omega = \ang{90}$, the antennas become polarization-mismatched.
Consequently, the channel gain is significantly reduced, and further adjustments of $\alpha$ and $\beta$ have a negligible impact on the gain, as shown in Fig.~\ref{fig:Hplane_sim_vs_mea}.

These results not only confirm the accuracy of the proposed EM-based channel model, but more importantly, they demonstrate the necessity of accounting for the \emph{vectorial} nature of EM propagation in channel modeling.
Specifically, existing works on SRAS-enabled systems often characterize antenna radiation using a \emph{scalar} function~\cite{10848372,11134688,zheng2025rotatab}. 
Channel models built upon such an approximation cannot distinguish between different levels of polarization alignment (e.g., as illustrated in Figs.~\ref{fig:Eplane_sim_vs_mea} and \ref{fig:Hplane_sim_vs_mea}), and thus are likely to significantly overestimate the channel gain. 
As will be shown in Sec. \ref{subsec:scalar_vs_vector_channel}, this inaccurate modeling can severely degrade the communication performance of SRAS-enabled systems.

\subsection{Transmit Antenna Orientation Design}

In this section, we evaluate the communication performance gain achieved by optimizing the orientation angle of each antenna in SRAS-enabled systems.
Specifically, we consider a downlink MU-MIMO system with $N_\mathrm{t}=8$ spatial reconfigurable transmit antennas spaced at half-wavelength intervals, while each communication user is equipped with $N_{\mathrm{u}}=2$ fixed receive antennas.
The orientations of transmit antennas are optimized by the proposed Riemannian gradient ascent method to maximize the sum-rate across all users.
To comprehensively assess the average performance gain of each scheme, we randomly generate $K$
 user positions within a distance range of 20$\lambda$ to 200$\lambda$ from the transmit antenna array, and then calculate the sum-rate.
This process is repeated 5000 times to determine the average sum-rate improvement of each scheme.
As a representative example, both transmit and receive antennas are modeled as half-wavelength dipoles.
The transmit power is set to $P_\mathrm{t} = 1$ W and noise power is $\sigma_{\mathrm{n}}^2 = -50$ dBm.
The convergence tolerance of the Riemannian gradient ascent is chosen as $\varepsilon = 10^{-5}$.

Fig. \ref{fig:comp_userNum} depicts the sum-rate versus the number of users $K$ in the considered MU-MIMO system.
The baseline schemes are:
1) a conventional fixed half-wavelength spacing ULA;
2) \textcolor{black}{movable antenna (MA)} arrays with different aperture sizes, where the antenna positions are determined by the greedy algorithm~\cite{10709885}.
As can be observed, for all schemes, the sum-rate consistently increases with increasing $K$, and the \textcolor{black}{rotatable antenna (RA)} scheme achieves significantly higher performance than the baseline schemes.
Fig. \ref{fig:comp_userNum} also demonstrates that the sum-rate of MAs approaches that of rotatable antennas only when the aperture size is sufficiently large and $K$ is close to $N_{\mathrm{t}}$.
The reason behind this is twofold.
First, when $K$ is small, MMSE precoder nearly eliminates the multiuser interference.
In this regime, the sum-rate is primarily limited by the signal-to-noise ratio (SNR).
For MAs, increasing the aperture size has little effect on SNR, which explains why the performance gap among MAs with different aperture sizes is negligible.
In contrast, as analyzed in Sec. \ref{subsec:impact_on_commRate}, RAs can enhance channel gain by adjusting the polarization alignment between transmit and receive antennas, thereby improving the SNR. 
Consequently, the RA system achieves a significantly higher sum-rate than all baseline schemes.
Second, as $K$ increases and approaches $N_{\mathrm{t}}$, users are more likely to be closer to each other, causing the channel matrix to become ill-conditioned~\cite[Ch.~10]{tseWirelessComm} (i.e., its condition number increases).
This leads to sum-rate saturation.
In this regime, the system performance is constrained by limited spatial DoFs.
For MA-enabled systems, enlarging the aperture size will enhance the spatial diversity and increase the available spatial DoFs.
This explains the notable performance gain achieved by large-aperture MAs.
Notably, in the considered scenario, the MAs must increase their aperture size by a factor of $2.5$ (i.e., from $4 \lambda$ to $10 \lambda$) to approach the sum-rate achieved by RAs in large-$K$ regime.

Fig.~\ref{fig:comp_SNR} shows the sum-rate versus the transmit power $P_{\mathrm{t}}$.
It can be observed that, for all schemes, the sum-rate increases monotonically with $P_{\mathrm{t}}$.
RAs exhibit the fastest growth rate, indicating that they deliver a greater sum-rate gain per unit of additional transmit power.
This is because RAs enhance the effective channel gain through polarization alignment (as analyzed in Sec. \ref{subsec:impact_on_commRate}), which directly amplifies the received SNR for a given $P_{\mathrm{t}}$.
Consequently, the benefit of increasing $P_{\mathrm{t}}$ is more pronounced for RAs than for the baseline schemes.
Specifically, when $P_{\mathrm{t}} = 20$~dBm, RAs and MAs improve the sum-rate by $9.2\%$ and $2.6\%$, respectively, compared to the fixed ULA.
When $P_{\mathrm{t}} = 70$~dBm, RAs achieve a $36\%$ gain, whereas MAs yield only $10\%$.

These results confirm that antenna orientation optimization, by improving polarization alignment, provides substantial sum-rate gains over both fixed-ULA and movable-antenna systems.

\begin{figure}[t]
	\centering
	\vspace{0pt}
	\hspace{-29pt}
	\includegraphics[trim=0 0 0 0 clip, width=0.32\textwidth]{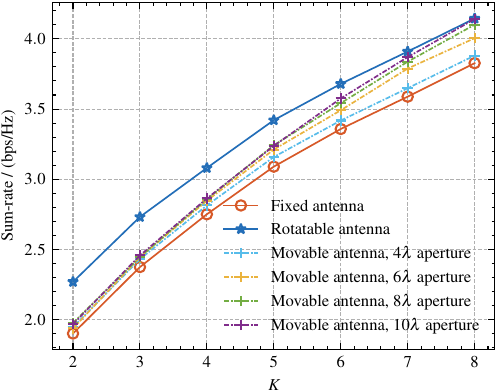}
	\caption{The sum-rate $R$ versus the number of users $K$ under different antenna systems.}
	\label{fig:comp_userNum}
	\vspace{0mm}  
\end{figure}

\begin{figure}[t]
	\centering
	\vspace{10pt}
	\hspace{-20pt}\includegraphics[trim=0 0 0 0 clip, width=0.32\textwidth]{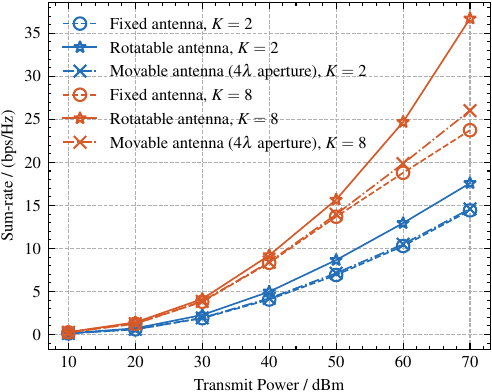}
	\vspace{0mm}
	\caption{The sum-rate $R$ versus the transmit power $P_{\mathrm{t}}$ under different antenna systems.}
	\label{fig:comp_SNR}
\end{figure}

\subsection{Necessity of the EM-based Channel Model for Orientation Design in RAs}\label{subsec:scalar_vs_vector_channel}

To demonstrate the necessity of the EM-based channel model for SRASs, in this section,
we compare the sum-rate achieved by RAs under different channel models.
Fig.~\ref{fig:scalar_vs_EM} compares three schemes:
1) orientations optimized using the proposed EM-based channel model;
2) orientations optimized using a widely adopted scalar channel model~\cite{10848372,11134688,zheng2025rotatab};
3) a conventional fixed antenna array.
Critically, the widely adopted scalar channel model characterizes the antenna radiation solely by a scalar function of the spatial angle, and therefore captures only the amplitude variation of the radiated field with direction.
It entirely omits the \emph{vectorial} polarization structure of the EM field.
{To ensure a fair comparison, the sum-rate of \emph{all} schemes is evaluated
under the proposed EM-based channel model}, whose accuracy has been independently
validated by both full-wave simulations and experimental measurements in
Sections~\ref{subsec:numerical_simulation} and~\ref{subsec:measurement}.

Two observations can be made from Fig.~\ref{fig:scalar_vs_EM}.
First, the orientation design under the EM-based channel model consistently achieves the highest sum-rate across all values of $K$.
Second, the scalar channel-based design not only yields a lower sum-rate
than the EM-based scheme, but as $K$ increases, it falls below even the conventional fixed antenna array.
This severe degradation stems from the incomplete physical modeling of the scalar channel model.
By describing antenna radiation using a scalar function and ignoring polarization mismatch, the scalar model produces biased predictions of channel gain under different antenna orientations.
As a result, orientations that appear optimal under the scalar model can, in the actual EM channel, create a strong polarization mismatch that significantly attenuates the received signal power.
This leads to a fundamental misdesign of the antenna orientations in RAs.
The observation is consistent with the measurement results in Figs.~\ref{fig:Eplane_sim_vs_mea} and~\ref{fig:Hplane_sim_vs_mea}, where changing the spin angle $\omega$ from $0^\circ$ to $90^\circ$ introduces a severe polarization mismatch that fundamentally alters the channel gain. This is an effect that the scalar model is structurally incapable of predicting.
These results demonstrate that the derived EM-based channel model is essential for unlocking the full potential of SRAS.

\begin{figure}[t]
	\centering
	\vspace{0pt}
	\hspace{0pt}\includegraphics[trim=0 0 0 0 clip, width=0.32\textwidth]{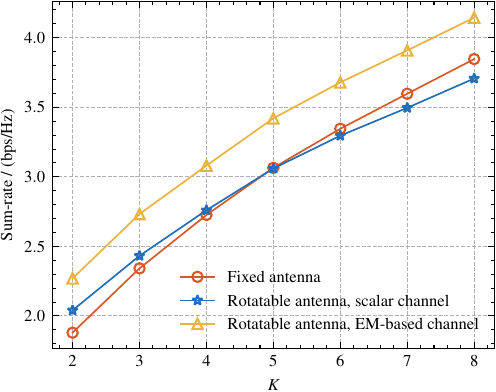}
	\vspace{0mm}
	\caption{The sum-rate $R$ versus the number of users $K$ under different channel models.}
	\label{fig:scalar_vs_EM}
\end{figure}

\subsection{Impact of Antenna Angular Resolution on Sum-rate}

In the previous simulations, we assumed that the orientation angles of RAs are continuous variables.
However, in practical systems, motor-driven rotators can typically adjust the antenna orientation only in discrete steps with a minimum step size, which inevitably degrades system performance.
In this section, we investigate the communication performance loss incurred due to this finite angular resolution of RAs.

Fig. \ref{fig:resolution} illustrates the degradation in communication performance as a function of the minimum rotation step size.
Specifically, given the sum-rate $R_{\Delta}$ of RAs with minimum rotation step size ${\Delta}$, its performance loss $L_{\Delta}$ is defined as
\begin{equation}
	\begin{aligned}
  		L_{\Delta} = \frac{R_{\mathrm{ideal}} - R_{\Delta} }{R_{\mathrm{ideal}} - R_{\mathrm{fix}}},
  	\end{aligned}
\end{equation} 
where $R_{\mathrm{ideal}}$ and $R_{\mathrm{fix}}$ denote the sum-rate of ideal RAs and conventional fixed antennas, respectively.
As can be observed, increasing the angular step size invariably degrades the system sum-rate, and this degradation becomes more pronounced as the number of users $K$ increases.
However, when $\Delta \leq \ang{10}$, the sum-rate loss of RAs is no more than $2\%$, and even when $\Delta$ increases to $\ang{20}$, the loss remains below $8\%$.
Notably, even the lowest-cost permanent-magnet motors can readily achieve a step size as fine as $\ang{7.5}$~\cite[Ch.~11]{gieras2002permanent}.
This implies that significant performance gains can be realized with RAs at minimal hardware cost, making RA-enabled communication systems attractive for practical deployment.

\begin{figure}[t]
	\centering
	\vspace{0pt}
	\hspace{0pt}\includegraphics[trim=0 0 0 0 clip, width=0.32\textwidth]{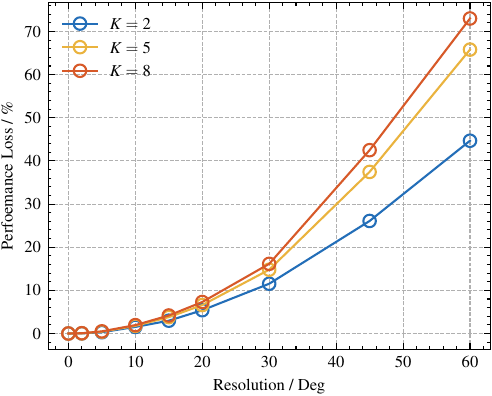}
	\vspace{0mm}
	\caption{The communication performance loss versus the minimum step size of antenna rotation angles.}
	\label{fig:resolution}
\end{figure}


\section{Conclusion}\label{sec:conclusion}

In this paper, we developed an EM-consistent channel modeling and orientation design framework for SRAS-enabled MIMO systems.
By leveraging the SVWE of the EM field, the proposed channel model applies to antennas of arbitrary structure and intrinsically captures polarization mismatch and other vectorial EM effects.
The model cleanly decouples antenna properties, orientation, and the propagation environment, enabling orientation-aware system design.
Both full-wave simulations and experimental measurements validated the accuracy of the proposed channel model.
The results further revealed that antenna orientation exerts a substantially stronger influence on the achievable rate than antenna displacement, primarily through polarization alignment.
By optimizing the orientation angles via the proposed Riemannian gradient ascent method, RAs achieved notable sum-rate gains over both movable-antenna and conventional fixed-antenna systems.
Moreover, the angular resolution analysis indicated that these gains are not severely degraded by the discrete step sizes of practical motor-driven rotators.
Based on these findings, we conclude that antenna orientation is an effective and practical design dimension for enhancing the communication performance of SRAS-enabled systems. This makes RAs a promising physical-layer technology for future 6G networks.

\appendices


\section{Proof of Proposition \ref{prop:hij}}\label{sec:appendix_prof_prop_hij}
The derivation follows the analysis of radiation, propagation and reception process presented in Sec. \ref{sec:EM_char_analysis}.

\paragraph{Radiation}

As illustrated in Fig.~\ref{fig:MIMO_setting_illu}, the transmit antenna is located at the origin of the coordinate system $(x,y,z)$.
When the transmit antenna is rotated, its radiation pattern in $(x,y,z)$ changes.
In the coordinate system $(x_{\mathrm{t}},y_{\mathrm{t}},z_{\mathrm{t}})$, the radiation field is given by
\begin{equation}\label{eq:appB_Et_svwe_in_Ot}
	\begin{aligned}
  		\mathbf{E}_{\mathrm{t}}(r_\mathrm{t},\theta_\mathrm{t},\varphi_\mathrm{t})=\frac{\kappa v_i}{\sqrt{\eta}}  \sum_{smn}   T_{smn}\mathbf{F}_{smn}^ {(3)} (r_\mathrm{t},\theta_\mathrm{t},\varphi_\mathrm{t}).
  	\end{aligned}
\end{equation}
Note that the coordinate transformation from $(x_\mathrm{t},y_\mathrm{t},z_\mathrm{t})$ to $(x,y,z)$ can be realized through the following sequence of operations:
\begin{itemize}
\item Rotate $(x_\mathrm{t},y_\mathrm{t},z_\mathrm{t})$ by an angle of $-\omega_{i}$ about the $z_\mathrm{t}$-axis to obtain $(x_1,y_1,z_1)$;
	\item Rotate $(x_1,y_1,z_1)$ by an angle of $-\beta_{i}$ about the $y_1$-axis to obtain $(x_2,y_2,z_2)$;
	\item Rotate $(x_2,y_2,z_2)$ by an angle of $-\alpha_{i}$ about the $z_2$-axis to obtain $(x,y,z)$.
\end{itemize}
Therefore, by leveraging the rotation property of SVWFs, we have
\begin{equation}\label{eq:appB_svwf_rot}
	\begin{aligned}
  		\mathbf{F}^{(3)}_{smn}(r_{\mathrm{t}},\theta_{\mathrm{t}},\varphi_{\mathrm{t}}) = \sum_{\mu=-n}^n  D^n_{\mu m}(-\omega_i,-\beta_i,-\alpha_i) \mathbf{F}^{(3)}_{s \mu n}(r,\theta,\varphi).
  	\end{aligned}
\end{equation}
Substituting (\ref{eq:appB_svwf_rot}) into (\ref{eq:appB_Et_svwe_in_Ot}) and rearranging terms yields the following expression of $\mathbf{E}_{\mathrm{t}}$,
\begin{equation}\label{eq:appB_Et_svwe_in_O}
	\begin{aligned}
  		\mathbf{E}_{\mathrm{t}}(r,\theta,\varphi)=\frac{\kappa v_i}{\sqrt{\eta}}  \sum_{s  \mu n} \left( \sum_{m=-n}^{n}  D^n_{\mu m}(-\omega_i,-\beta_i,-\alpha_i) T_{smn}  \right) \\ \mathbf{F}_{s \mu n}^ {(3)} (r,\theta,\varphi).
  	\end{aligned}
\end{equation}
(\ref{eq:appB_Et_svwe_in_O}) shows that, in the coordinate system $(x,y,z)$, the radiation coefficient $T'_{s \mu n}$ of the rotated antenna is given by
\begin{equation}\label{eq:appB_radi_coeff_in_O}
	\begin{aligned}
  		T'_{s \mu n} = \sum_{m=-n}^{n} D^n_{\mu m}(-\omega_i,-\beta_i,-\alpha_i) T_{smn}.
  	\end{aligned}
\end{equation}
(\ref{eq:appB_radi_coeff_in_O}) can be further expressed in a matrix form as
\begin{equation}\label{eq:appB_t'}
	\begin{aligned}
  		\mathbf{t}'=\mathbf{O}_i\mathbf{t}.
  	\end{aligned}
\end{equation}
The matrix $\mathbf{O}_i$ depends on the orientation of the $i$-th transmit antenna.
Using the index mappings $(s,\mu,n) \rightarrow p$ and $(\sigma,\rho,\nu) \rightarrow q$, the $(p,q)$-th element of $\mathbf{O}_i$ is given by
\begin{equation}
	\begin{aligned}
   		[\mathbf{O}_i]_{p,q} =  \delta_{s,\sigma}\,\delta_{n,\nu}\,D^{n}_{\mu \rho}(-\omega_i,-\beta_i,-\alpha_i).
  	\end{aligned}
\end{equation}
\textcolor{black}{Note that $\mathbf{O}_i$ is block-diagonal in the indices $s$ and $n$, i.e., only entries where $p$ and $q$ share the same $s$ and $n$ are non-zero.}

\paragraph{Propagation}

To characterize the propagation process, we first examine the transformation relationship of a single SVWF $\mathbf{F}_{s \mu n}^{(3)}$ between the two coordinate systems $(x,y,z)$ and $(x_\mathrm{r},y_\mathrm{r},z_\mathrm{r})$ depicted in Fig. \ref{fig:MIMO_setting_illu}.
From a spatial perspective, the transformation from $(x,y,z)$ to $(x_\mathrm{r},y_\mathrm{r},z_\mathrm{r})$ can be realized through the following sequence of operations:
\begin{itemize}
	\item Rotate $(x,y,z)$ by an angle of $\varphi_{ji}$ about the $z$-axis to obtain $(x_1,y_1,z_1)$;
	\item Rotate $(x_1,y_1,z_1)$ by an angle of $\theta_{ji}$ about the $y_1$-axis to obtain $(x_2,y_2,z_2)$;
	\item Translate $(x_2,y_2,z_2)$ by a distance $\Vert \mathbf{d}_{ji} \Vert$ along the $z_2$-axis to obtain $(x_3,y_3,z_3)$;
	\item Rotate $(x_3,y_3,z_3)$ by an angle of $-\theta_{ji}$ about the $y_3$-axis to obtain $(x_4,y_4,z_4)$;
	\item Rotate $(x_4,y_4,z_4)$ by an angle of $-\varphi_{ji}$ about the $z_4$-axis to arrive at the target coordinate system $(x_\mathrm{r},y_\mathrm{r},z_\mathrm{r})$;
\end{itemize}
By leveraging the rotation and translation property of SVWFs, we have
\begin{equation}\label{eq:Fsmn_general_transform}
	\begin{aligned}
		\mathbf{F}_{s \mu n}^{(3)}  (r,\theta,\varphi) =  \sum_{\sigma \nu  \gamma \rho}	 D^{n}_{\gamma \mu}({  0},\theta_{ji},{  \varphi_{ji}})  D^{\nu}_{\rho \gamma}({  -\varphi_{ji}},-\theta_{ji},{  0}) \\
		C^{sn(3)}_{\sigma \gamma \nu}(\kappa \Vert \mathbf{d}_{ji} \Vert) \mathbf{F}_{\sigma \rho \nu}^{(1)}(r_\mathrm{r},\theta_\mathrm{r},\varphi_\mathrm{r}).
	\end{aligned}
\end{equation}
Substituting (\ref{eq:Fsmn_general_transform}) into (\ref{eq:appB_Et_svwe_in_O}), we obtain the SVWE of the radiation field in the coordinate system $(x_\mathrm{r},y_\mathrm{r},z_\mathrm{r})$, which is given by
\begin{equation}\label{eq:Et_in_Or}
    	\begin{aligned}
    		\mathbf{E}_{\mathrm{t}} (r_\mathrm{r},\theta_\mathrm{r},\varphi_\mathrm{r}) = 
  			\frac{\kappa v_i}{\sqrt{\eta}}	\sum_{\substack {s \mu n\\ \sigma \nu  \gamma \rho}} T'_{s \mu n}  D^{n}_{\gamma \mu}({  0},\theta_{ji},{  \varphi_{ji}}) \\  D^{\nu}_{\rho \gamma}({  -\varphi_{ji}},-\theta_{ji},{  0})
		C^{sn(3)}_{\sigma \gamma \nu}(\kappa \Vert \mathbf{d}_{ji} \Vert) \mathbf{F}_{\sigma \rho \nu}^{(1)}(r_\mathrm{r},\theta_\mathrm{r},\varphi_\mathrm{r})
  		\end{aligned}
  \end{equation}
It can be observed that in the coordinate system $(x_\mathrm{r},y_\mathrm{r},z_\mathrm{r})$, the radiated field of the transmit antenna can be represented as a superposition of infinitely many standing waves.
Note that each standing wave can be decomposed into two counter-propagating traveling waves, i.e.,
\begin{equation}
	\begin{aligned}
  		\mathbf{F}_{\sigma \rho \nu}^{(1)}(r_\mathrm{r},\theta_\mathrm{r},\varphi_\mathrm{r}) =\frac{1}{2}\left(\mathbf{F}_{\sigma \rho \nu}^{(3)}(r_\mathrm{r},\theta_\mathrm{r},\varphi_\mathrm{r}) + \mathbf{F}_{\sigma \rho \nu}^{(4)}(r_\mathrm{r},\theta_\mathrm{r},\varphi_\mathrm{r})\right), \notag
  	\end{aligned}
\end{equation}
where $\mathbf{F}_{\sigma \rho \nu}^{(3)}$ and $\mathbf{F}_{\sigma \rho \nu}^{(4)}$ correspond to the outward- and inward-propagating components, respectively.
Noting that only the inward-propagating wave $\mathbf{F}_{\sigma \rho \nu}^{(4)}(r_\mathrm{r},\theta_\mathrm{r},\varphi_\mathrm{r})$ contributes to the response at the port of the receive antenna, the effective incident electric field $\mathbf{E}_{\mathrm{r}}$ at the receiver is given by
\begin{equation}\label{eq:Er_eff_in_Or}
	\begin{aligned}
  		\mathbf{E}_{\mathrm{r}} (r_\mathrm{r},\theta_\mathrm{r},\varphi_\mathrm{r}) = 
  			\frac{\kappa v_i}{2 \sqrt{\eta}}	\sum_{\substack {s \mu n\\ \sigma \nu  \gamma \rho}} T'_{s \mu n}  D^{n}_{\gamma \mu}({  0},\theta_{ji},{  \varphi_{ji}}) \\  D^{\nu}_{\rho \gamma}({  -\varphi_{ji}},-\theta_{ji},{  0})
		C^{sn(3)}_{\sigma \gamma \nu}(\kappa \Vert \mathbf{d}_{ji} \Vert) \mathbf{F}_{\sigma \rho \nu}^{(4)}(r_\mathrm{r},\theta_\mathrm{r},\varphi_\mathrm{r}).
  	\end{aligned}
\end{equation}

\paragraph{Reception}

According to (\ref{eq:Er_receive_signal_form}), the received signal $w_j$ at the receive antenna's port can be expressed as
\begin{equation}\label{eq:wj_appendix}
	\begin{aligned}
  		w_j =  
  		  \frac{\kappa v_i}{2 \sqrt{\eta}}	
  		  \sum_{\substack {s \mu n\\ \sigma \nu  \gamma \rho}} 
  		  T'_{s \mu n} R_{\sigma \rho \nu} 
  		  D^{n}_{\gamma \mu}({  0},\theta_{ji},{  \varphi_{ji}}) 
  		  \\
  		  D^{\nu}_{\rho \gamma}({  -\varphi_{ji}},-\theta_{ji},{  0}) 
  		  C^{sn(3)}_{\sigma \gamma \nu}(\kappa \Vert \mathbf{d}_{ji} \Vert)
  	\end{aligned}
\end{equation}

Based on the SVWE framework, we have fully characterized the signal transmission process in SRASs.
By applying index mapping $(s,\mu,n) \leftrightarrow p$ and $(\sigma,\rho,\nu) \leftrightarrow q$ to (\ref{eq:wj_appendix}), the channel gain $h_{ji} = w_j / v_i$ between the transmit and receive antenna can be expressed in a matrix form as
\begin{equation}\label{eq:hij_appendix}
	\begin{aligned}
  		h_{ji} = \mathbf{r}^T {\hat{\mathbf{D}}}_{ji} \mathbf{t}'.
  	\end{aligned}
\end{equation}
The matrix $\hat{\mathbf{D}}_{ji}$ only depends on the displacement vector $\mathbf{d}_{ji}$, and its $({  q},{  p})$-th element is given by
\begin{equation}
	\begin{aligned}\label{eq:gene_G}
  		[\hat{{\mathbf{D}}}_{ji}]_{{  q},{  p}} = \frac{\kappa}{2\sqrt{\eta}}   
  		&\sum_{\gamma=-n}^{n}
  		D^{n}_{\gamma \mu}({  0},\theta_{ji},{  \varphi_{ji}}) 
  		\\
  		&D^{\nu}_{\rho \gamma}({  -\varphi_{ji}},-\theta_{ji},{  0}) 
  		C^{sn(3)}_{\sigma \gamma \nu}(\kappa \Vert \mathbf{d}_{ji} \Vert)
  	\end{aligned}
\end{equation}
Under the far-field assumption, i.e., $\kappa \Vert \mathbf{d}_{ji} \Vert \rightarrow \infty$, the following asymptotic formulations hold:
\begin{subequations}
	\begin{align}
		&	C^{sn(3)}_{\sigma \gamma \nu}(A) = o(\frac{1}{A}) \quad \text{for }\gamma \neq \pm 1, \\
		&	C^{sn(3)}_{\sigma 1 \nu}(A) = a(n,\nu) \frac{e^{\jmath A}}{A}+o(\frac{1}{A}), \\
		&	C^{sn(3)}_{\sigma, -1, \nu}(A) = (-1)^{s+\sigma} a(n,\nu) \frac{e^{\jmath A}}{A}+o(\frac{1}{A}),
	\end{align}
\end{subequations}
where $a(n,\nu)$ is defined in (\ref{eq:anv}).
This implies that (\ref{eq:gene_G}) can be further simplified to
\begin{equation}
	\begin{aligned}
  		[\hat{{\mathbf{D}}}_{ji}]_{{  q},{  p}} = 
  		\frac{a(n,\nu)}{2\sqrt{\eta}}
  		\frac{e^{\jmath \kappa \Vert \mathbf{d}_{ji} \Vert}}{ \Vert \mathbf{d}_{ji} \Vert} 
  		e^{\jmath ( \mu - \rho)\varphi_{ji}} 
  		b^{s,\mu,n}_{\sigma,\rho,\nu}(\theta_{ji})
  	\end{aligned}
\end{equation}
where $a(n,\nu)$ and $b^{s,\mu,n}_{\sigma,\rho,\nu}(\theta_{ji})$ are defined in (\ref{eq:anv}) and (\ref{eq:b_theta}), respectively.
We define $\mathbf{P}_{ji}$ to represent the part of $\hat{\mathbf{D}}_{ji}$ that depends on the relative angles $\varphi_{ji}$ and $\theta_{ji}$ between the $i$-th transmit and $j$-th receive antenna. Substituting  (\ref{eq:appB_t'}) into (\ref{eq:hij_appendix}) yields the channel gain expression presented in (\ref{eq:EM_based_SISO_channel_model}), thereby completing the proof.

\section{Derivation of the Riemannian Gradient Components}\label{sec:appendix_riemannian_gradient}

\textcolor{black}{We first show that ${\mathbf{G}}_i$ admits the expansion ${\mathbf{G}}_i = g_{i,x}\mathbf{G}_x + g_{i,y}\mathbf{G}_y + g_{i,z}\mathbf{G}_z$. Since $\operatorname{grad}_{\mathbf{U}_i} R$ is a tangent vector in $T_{\mathbf{U}_i}\mathrm{SO}(3)$, it must take the form $\mathbf{U}_i\,{\mathbf{G}}_i$ where ${\mathbf{G}}_i$ is a $3\times 3$ skew-symmetric matrix. The matrices $\mathbf{G}_x$, $\mathbf{G}_y$, and $\mathbf{G}_z$ form the standard basis of $3\times 3$ skew-symmetric matrices, and any such matrix admits a unique expansion in this basis, yielding (\ref{eq:hat_gi}). The main task is therefore to derive the explicit expressions of the coefficients $g_{i,a}$ ($a\in\{x,y,z\}$), which are the directional derivatives $dR/dt|_{t=0}$ under the infinitesimal rotation $\mathbf{U}_i \to \mathbf{U}_i \cdot \exp(t \mathbf{G}_a)$. These coefficients are obtained via the chain rule through the EM channel model, as detailed in the following subsections.}

\subsection{Derivative of the Orientation Matrix}

Recall that the entries of $\mathbf{O}_i$ are the Wigner D-functions $D^n_{\mu m}$.
In the language of group theory, this means that $\mathbf{O}_i$ is a \emph{representation} of $\mathrm{SO}(3)$ on the SVWF space. The representation is a matrix-valued function that maps each rotation to a linear operator and respects the group structure.
A key consequence of this structure is that the derivative of $\mathbf{O}_i$ with respect to an infinitesimal body-axis rotation admits a compact closed-form expression involving only the known and precomputable operators $\mathrm{d}\rho(\mathbf{G}_a)$, namely
\begin{equation}
	\begin{aligned}
  		\left.\frac{\mathrm{d}}{\mathrm{d}t}\mathbf{O}_i \bigl(\mathbf{U}_i \cdot \exp{(t \mathbf{G}_a)}\bigr)\right|_{t=0} = \mathbf{O}_i  \mathrm{d}\rho(\mathbf{G}_a),
  	\end{aligned}
\end{equation}
where the operators $\mathrm{d}\rho(\mathbf{G}_a)$ are simply $L\times L$ block-diagonal matrices acting on the SVWF mode coefficients.
Specifically, under the index mapping $(s,\mu,n)\leftrightarrow p$ and $(s,m,n)\leftrightarrow q$, for each degree $n$, $(\mu,m)$-th entries of $\mathrm{d}\rho(\mathbf{G}_a)$ are
\begin{align}
  		\bigl[\mathrm{d}\rho(\mathbf{G}_z)\bigr]_{\mu,m} &= -\jmath m\,\delta_{\mu,m},\\
	\bigl[\mathrm{d}\rho(\mathbf{G}_x)\bigr]_{\mu,m} &= -\frac{\jmath}{2}\bigl(c_{n,m}^+\,\delta_{\mu,m+1} + c_{n,m}^-\,\delta_{\mu,m-1}\bigr),\\
	\bigl[\mathrm{d}\rho(\mathbf{G}_y)\bigr]_{\mu,m} &= -\frac12\bigl(c_{n,m}^+\,\delta_{\mu,m+1} - c_{n,m}^-\,\delta_{\mu,m-1}\bigr),
  	\end{align}
with $c_{n,m}^{\pm} = \sqrt{n(n+1) - m(m\pm 1)}$.

\subsection{Perturbation of the Channel Matrix}

Let $\tilde{\mathbf{H}}^{(a,i)} = \partial\mathbf{H}/\partial t|_{t=0}$.
Only column $i$ is non-zero, with entries
\begin{equation}
	\begin{aligned}
  		\bigl[\tilde{\mathbf{H}}^{(a,i)}\bigr]_{:,i}
	= \frac{e^{\jmath \kappa\|\mathbf{d}_{ji}\|}}{\kappa\|\mathbf{d}_{ji}\|}\;
   \mathbf{r}^T \mathbf{P}_{ji}\; \mathbf{O}_i\,d\rho(\mathbf{G}_a)\,\mathbf{t}.
  	\end{aligned}
\end{equation}

\subsection{Perturbation of the MMSE Precoder}

The MMSE precoder is $\mathbf{F} = \zeta \mathbf{F}_{\mathrm{un}}$ with $\mathbf{F}_{\mathrm{un}} = \mathbf{H}^H\mathbf{M}$ and $\mathbf{M} = (\mathbf{H}\mathbf{H}^H + \epsilon\mathbf{I})^{-1}$.
The perturbation of $\mathbf{F}_{\mathrm{un}}$ induced by the rotation of antenna $i$ about axis $a$ is $\tilde{\mathbf{F}}_{\mathrm{un}}^{(a,i)} = \partial\mathbf{F}_{\mathrm{un}}/\partial t|_{t=0}$, which depends on $\tilde{\mathbf{H}}^{(a,i)}$ through the matrix derivative of $\mathbf{M}$:
\begin{equation}
	\begin{aligned}
  		\tilde{\mathbf{F}}_{\mathrm{un}}^{(a,i)}
= \bigl(\mathbf{I}_{N_{\mathrm{t}}} - \mathbf{H}^H \mathbf{M} \mathbf{H}\bigr)\,
   \bigl(\tilde{\mathbf{H}}^{(a,i)}\bigr)^H \mathbf{M}
   - \mathbf{H}^H \mathbf{M}\, \tilde{\mathbf{H}}^{(a,i)}\, \mathbf{H}^H \mathbf{M}.
  	\end{aligned}
\end{equation}
The power normalization factor is $\zeta = \sqrt{P}/\|\mathbf{F}_{\mathrm{un}}\|_F$, with derivative
$\tilde{\zeta}^{(a,i)} = -\zeta  \mathrm{Re}\,\mathrm{Tr}\bigl(\mathbf{F}_{\mathrm{un}}^H \tilde{\mathbf{F}}_{\mathrm{un}}^{(a,i)}\bigr) / \|\mathbf{F}_{\mathrm{un}}\|_F^2$.
The full precoder perturbation is then
\begin{equation}
	\begin{aligned}
  		\tilde{\mathbf{F}}^{(a,i)} = \zeta\bigl(\tilde{\mathbf{F}}_{\mathrm{un}}^{(a,i)} + \tilde{\zeta}_0^{(a,i)}\,\mathbf{F}_{\mathrm{un}}\bigr),
  	\end{aligned}
\end{equation}
where $\tilde{\zeta}_0^{(a,i)} = \tilde{\zeta}^{(a,i)}/\zeta$.
Partitioning 
\begin{equation}
	\begin{aligned}
  		\tilde{\mathbf{F}}^{(a,i)} = [\tilde{\mathbf{F}}_1^{(a,i)},\dots,\tilde{\mathbf{F}}_K^{(a,i)}]
  	\end{aligned}
\end{equation}
with $\tilde{\mathbf{F}}_j^{(a,i)} \in \mathbb{C}^{N_{\mathrm{t}}\times N_{\mathrm{u}}}$,
the $\mathbf{Q}$-matrix perturbation is
\begin{equation}
	\begin{aligned}
  		\tilde{\mathbf{Q}}_j^{(a,i)} = \tilde{\mathbf{F}}_j^{(a,i)} \mathbf{F}_j^H + \mathbf{F}_j \bigl(\tilde{\mathbf{F}}_j^{(a,i)}\bigr)^H.
  	\end{aligned}
\end{equation}

\vspace{-10pt}

\subsection{Derivative of the Per-User Rate}

Define $\mathbf{A}_k = \mathbf{H}_k \mathbf{Q}_k \mathbf{H}_k^H \mathbf{N}_k$, $\mathbf{B}_k = (\mathbf{I} + \mathbf{A}_k)^{-1}$, and the Hermitian matrix
\begin{equation}
	\begin{aligned}
  		\tilde{\mathbf{C}}_{kj}^{(a,i)}
= \tilde{\mathbf{H}}_k \mathbf{Q}_j \mathbf{H}_k^H
 + \mathbf{H}_k \mathbf{Q}_j \tilde{\mathbf{H}}_k^H
 + \mathbf{H}_k \tilde{\mathbf{Q}}_j^{(a,i)} \mathbf{H}_k^H,
  	\end{aligned}
\end{equation}
where $\tilde{\mathbf{H}}_k$ is the block of rows belonging to user $k$ of $\tilde{\mathbf{H}}^{(a,i)}$.
The directional derivative of $R_k$ is then
\begin{equation}
	\begin{aligned}
  		\left.\frac{dR_k}{dt}\right|_{t=0}
= \mathrm{Tr} \left( \mathbf{B}_k (\tilde{\mathbf{C}}_{kk}^{(a,i)} - \mathbf{A}_k \sum_{j\neq k}\tilde{\mathbf{C}}_{kj}^{(a,i)})  \mathbf{N}_k \right).
  	\end{aligned}
\end{equation}
Summing $dR_k/dt$ over all $K$ users yields the required gradient component $g_{i,a}$.

\bibliographystyle{IEEEtran}
\bibliography{IEEEabrv,references}

@STRING{TSP         = "{IEEE} Trans. Signal Process."}

@STRING{IEEE_CL       = "{IEEE} Commun. Lett."}

@STRING{JSAC       = "{IEEE} J. Sel. Areas Commun."}

@STRING{TCOMM        = "{IEEE} Trans. Commun."}

@STRING{IEEE_TWC       = "{IEEE} Trans. Wireless Commun."}

@STRING{TWC       = "{IEEE} Trans. Wireless Commun."}

@STRING{IEEE_CST = "IEEE Commun. Surv. Tutor."}

@STRING{IEEE_APMag = "IEEE Antennas Propag. Mag."}

@STRING{IEEE_OJAP= "IEEE Open J. Antennas Propag."}

@STRING{TSP       = "{IEEE} Trans. Signal Process."}

@STRING{TAP       = "{IEEE} Trans. Antennas Propag."}

@article{xu2026generalembasedchannelmodel,
	title={A General {EM}-Based Channel Model for Reconfigurable Antenna Systems}, 
	author={Chen Xu and Xianghao Yu},
	year={2026},
	month={Apr.},
	journal={arXiv preprint arXiv:2604.22264}, 
}

@article{Song2025,
  title = {An overview on IRS-enabled sensing and communications for 6G: architectures,  fundamental limits,  and joint beamforming designs},
  volume = {68},
  ISSN = {1869-1919},
  DOI = {10.1007/s11432-024-4345-5},
  number = {5},
  journal = {Sci. China Inf. Sci.},
  publisher = {Springer Science and Business Media LLC},
  author = {Song,  Xianxin and Fang,  Yuan and Wang,  Feng and Ren,  Zixiang and Yu,  Xianghao and Zhang,  Ye and Liu,  Fan and Xu,  Jie and Ng,  Derrick Wing Kwan and Zhang,  Rui and Cui,  Shuguang},
  year = {2025},
  month = Apr 
}

@ARTICLE{1445988,
  author={Waterman, P.C.},
  journal={Proc. IEEE}, 
  title={Matrix formulation of electromagnetic scattering}, 
  year={1965},
  volume={53},
  number={8},
  pages={805-812},
  month={Aug.},
  doi={10.1109/PROC.1965.4058}}

@ARTICLE{6231655,
  author={Della Giovampaola, Cristian and Martini, Enrica and Toccafondi, Alberto and Maci, Stefano},
  journal=TAP, 
  title={A Hybrid {PO}/{G}eneralized-Scattering-Matrix Approach for Estimating the Reflector Induced Mismatch}, 
  year={2012},
  volume={60},
  number={9},
  month={Jul.},
  pages={4316-4325},
  doi={10.1109/TAP.2012.2207062}}

@ARTICLE{1406246,
  author={Rubio, J. and Gonzalez, M.A. and Zapata, J.},
  journal=TAP,
  title={Generalized-scattering-matrix analysis of a class of finite arrays of coupled antennas by using 3-{D} {FEM} and spherical mode expansion}, 
  year={2005},
  volume={53},
  number={3},
  pages={1133-1144},
  month={Mar.},
  doi={10.1109/TAP.2004.842687}}

@article{Gouesbet2010,
  title = {T-matrix formulation and generalized {L}orenz-{M}ie theories in spherical coordinates},
  volume = {283},
  ISSN = {0030-4018},
  DOI = {10.1016/j.optcom.2009.10.092},
  number = {4},
  journal = {Opt. Commun.},
  publisher = {Elsevier BV},
  author = {Gouesbet,  G.},
  year = {2010},
  month = {Feb.},
  pages = {517-521}
}

@ARTICLE{11142311,
  author={Shao, Xiaodan and Mei, Weidong and You, Changsheng and Wu, Qingqing and Zheng, Beixiong and Wang, Cheng-Xiang and Li, Junling and Zhang, Rui and Schober, Robert and Zhu, Lipeng and Zhuang, Weihua and Shen, Xuemin},
  journal=IEEE_CST, 
  title={A Tutorial on Six-Dimensional Movable Antenna for {6G} Networks: Synergizing Positionable and Rotatable Antennas}, 
  year={2026},
  volume={28},
  number={},
  month={Aug.},
  pages={3666-3709},
  doi={10.1109/COMST.2025.3602939}}

@ARTICLE{9264694,
  author={Wong, Kai-Kit and Shojaeifard, Arman and Tong, Kin-Fai and Zhang, Yangyang},
  journal=TWC, 
  title={Fluid Antenna Systems}, 
  year={2021},
  volume={20},
  number={3},
  month={Nov.},
  pages={1950-1962},
  doi={10.1109/TWC.2020.3037595}}

@ARTICLE{1261332,
  author={Spencer, Q.H. and Swindlehurst, A.L. and Haardt, M.},
  journal=TSP, 
  title={Zero-forcing methods for downlink spatial multiplexing in multiuser {MIMO} channels}, 
  year={2004},
  volume={52},
  number={2},
  month={Feb.},
  pages={461-471},
  doi={10.1109/TSP.2003.821107}}

@ARTICLE{9650760,
  author={Wong, Kai-Kit and Tong, Kin-Fai},
  journal=TWC, 
  title={Fluid Antenna Multiple Access}, 
  year={2022},
  volume={21},
  number={7},
  month={Dec.},
  pages={4801-4815},
  doi={10.1109/TWC.2021.3133410}}

@ARTICLE{10500751,
  author={Gong, Tierui and Wei, Li and Huang, Chongwen and Yang, Zhijia and He, Jiguang and Debbah, Mérouane and Yuen, Chau},
  journal=JSAC, 
  title={Holographic {MIMO} Communications With Arbitrary Surface Placements: Near-Field {LoS} Channel Model and Capacity Limit}, 
  year={2024},
  volume={42},
  number={6},
  pages={1549-1566},
  month = {Apr.},
  doi={10.1109/JSAC.2024.3389126}}

@ARTICLE{11258091,
  author={Zhang, Runxin and Shao, Yulin and Eldar, Yonina C.},
  journal=TWC, 
  title={Polarization-Aware Movable Antenna}, 
  year={2026},
  volume={25},
  number={},
  pages={7428-7442},
  month={Nov.},
  doi={10.1109/TWC.2025.3631482}}

@ARTICLE{11006094,
  author={Pizzo, Andrea and Marzetta, Thomas L. and Sanguinetti, Luca},
  journal=JSAC, 
  title={Spatially-Stationary Model for Holographic {MIMO} Small-Scale Fading},
  year={2020},
  volume={38},
  number={9},
  pages={1964-1979},
  month={Jun.},
  doi={10.1109/JSAC.2020.3000877}}

@book{Hansen_1988, title={Spherical near-field antenna measurements}, publisher={Institution of Engineering and Technology - IET}, author={Hansen, J. E.}, year={1988}}

@ARTICLE{10500425,
  author={Zhi, Kangda and Pan, Cunhua and Ren, Hong and Chai, Kok Keong and Wang, Cheng Xiang and Schober, Robert and You, Xiaohu},
  journal=JSAC, 
  title={Performance Analysis and Low-Complexity Design for {XL-MIMO} With Near-Field Spatial Non-Stationarities}, 
  year={2024},
  volume={42},
  number={6},
  month={Apr.},
  pages={1656-1672},
  doi={10.1109/JSAC.2024.3389128}}

@ARTICLE{liu2025ma,
  author={Liu, Shicong and Yu, Xianghao and Xu, Jie and Zhang, Rui},
  journal=IEEE_TWC, 
  title={Near-Field Communication With Massive Movable Antennas: {A} Functional Perspective}, 
  year={2026},
  volume={25},
  month={Apr.},
  number={},
  pages={14455-14470},
  doi={10.1109/TWC.2026.3677546}}

@ARTICLE{liu2025nfddd,
  author={Liu, Shicong and Yu, Xianghao and Song, Shenghui and Letaief, Khaled B.},
  journal=IEEE_TWC, 
  title={Near-Field Communication With Movable Antennas: An Electrostatic Equilibrium Perspective}, 
  year={2026},
  volume={25},
  number={},
  month={Apr.},
  pages={16024-16039},
  doi={10.1109/TWC.2026.3686424}}

@inbook{gieras2002permanent,
  title={Permanent Magnet Motor Technology: Design and Applications},
  author={Gieras, J.F. and Wing, M.},
  chapter = {11},
  isbn={9780824743949},
  year={2002},
  pages={457--484},
  publisher={Taylor \& Francis},
  address={Abingdon, UK}
}

@book{Stratton2015,
  title = {Electromagnetic Theory},
  ISBN = {9781119134640},
  DOI = {10.1002/9781119134640},
  publisher = {Wiley},
  author = {Stratton,  Julius Adams},
  year = {2015},
  month = {Oct.} 
}

@ARTICLE{1143727,
  author={Yaghjian, A.},
  journal=TAP, 
  title={An overview of near-field antenna measurements}, 
  year={1986},
  volume={34},
  number={1},
  pages={30-45},
  month={Jan.},
  doi={10.1109/TAP.1986.1143727}}

@inbook{tseWirelessComm,
author = {Tse, David and Viswanath, Pramod},
title = {Fundamentals of wireless communication},
year = {2005},
isbn = {0521845270},
chapter={10},
pages={497--578},
publisher = {Cambridge University Press},
address={New York, USA}
}

@ARTICLE{9110848,
  author={Pizzo, Andrea and Marzetta, Thomas L. and Sanguinetti, Luca},
  journal=JSAC, 
  title={Spatially-Stationary Model for Holographic {MIMO} Small-Scale Fading}, 
  year={2020},
  volume={38},
  number={9},
  month={Jun.},
  pages={1964-1979},
  doi={10.1109/JSAC.2020.3000877}}

@book{alma990026142230205776,
  title = {Absorption and Scattering of Light by Small Particles},
  ISBN = {9783527618156},
  DOI = {10.1002/9783527618156},
  publisher = {Wiley},
  author = {Bohren,  Craig F. and Huffman,  Donald R.},
  year = {1998},
  month = {Apr.},
}

@ARTICLE{6474484,
  author={Haupt, Randy L. and Lanagan, Michael},
  journal=IEEE_APMag, 
  title={Reconfigurable Antennas}, 
  year={2013},
  volume={55},
  number={1},
  pages={49-61},
  month={Mar.},
  doi={10.1109/MAP.2013.6474484}}

@ARTICLE{10740058,
  author={Zhang, Jichen and Rao, Junhui and Li, Zan and Ming, Zhaoyang and Chiu, Chi-Yuk and Wong, Kai-Kit and Tong, Kin-Fai and Murch, Ross},
  journal=IEEE_OJAP, 
  title={A Novel Pixel-Based Reconfigurable Antenna Applied in Fluid Antenna Systems With High Switching Speed}, 
  year={2025},
  volume={6},
  number={1},
  pages={212-228},
  month={Nov.},
  doi={10.1109/OJAP.2024.3489215}}

@ARTICLE{9131873,
  author={Wong, Kai Kit and Shojaeifard, Arman and Tong, Kin-Fai and Zhang, Yangyang},
  journal=IEEE_CL, 
  title={Performance Limits of Fluid Antenna Systems}, 
  year={2020},
  volume={24},
  number={11},
  pages={2469-2472},
  month={Jul.},
  doi={10.1109/LCOMM.2020.3006554}}

@ARTICLE{zheng2025rotatab,
  author={Zheng, Beixiong and Wu, Qingjie and Ma, Tiantian and Zhang, Rui},
  journal=TCOMM, 
  title={Rotatable Antenna-Enabled Wireless Communication: Modeling and Optimization}, 
  year={2026},
  volume={74},
  month={Mar.},
  number={},
  pages={6825-6842},
  doi={10.1109/TCOMM.2026.3672230}}

@ARTICLE{10318061,
  author={Zhu, Lipeng and Ma, Wenyan and Zhang, Rui},
  journal=TWC, 
  title={Modeling and Performance Analysis for Movable Antenna Enabled Wireless Communications}, 
  year={2024},
  volume={23},
  number={6},
  month={Nov.},
  pages={6234-6250},
  doi={10.1109/TWC.2023.3330887}}

@ARTICLE{10709885,
  author={Zhu, Lipeng and Ma, Wenyan and Xiao, Zhenyu and Zhang, Rui},
  journal=TWC, 
  title={Performance Analysis and Optimization for Movable Antenna Aided Wideband Communications}, 
  year={2024},
  volume={23},
  number={12},
  pages={18653-18668},
  month={Oct.},
  doi={10.1109/TWC.2024.3471698}}

@book{absil2009optimization,
  title={Optimization algorithms on matrix manifolds},
  author={Absil, P-A and Mahony, Robert and Sepulchre, Rodolphe},
  year={2009},
  publisher={Princeton University Press}
}

@book{Hall2003,
  title = {Lie Groups,  Lie Algebras,  and Representations},
  ISBN = {9780387215549},
  ISSN = {0072-5285},
  DOI = {10.1007/978-0-387-21554-9},
  journal = {Graduate Texts in Mathematics},
  publisher = {Springer New York},
  author = {Hall,  Brian C.},
  year = {2003}
}

@ARTICLE{11134688,
  author={Xiong, Xue and Zheng, Beixiong and Wu, Wen and Shao, Xiaodan and Dai, Liang and Zhao, Ming-Min and Tang, Jie},
  journal=IEEE_CL, 
  title={Efficient Channel Estimation for Rotatable Antenna-Enabled Wireless Communication}, 
  year={2025},
  volume={14},
  number={11},
  month={Aug.},
  pages={3719-3723},
  doi={10.1109/LWC.2025.3601979}}

@ARTICLE{10848372,
  author={Shao, Xiaodan and Zhang, Rui and Jiang, Qijun and Schober, Robert},
  journal=JSAC, 
  title={{6D} Movable Antenna Enhanced Wireless Network via Discrete Position and Rotation Optimization}, 
  year={2025},
  volume={43},
  month={Jan},
  number={3},
  pages={674-687},
  doi={10.1109/JSAC.2025.3531571}}

@ARTICLE{Yang2017,
  author={Yang, Xue and Xu, Shenheng and Yang, Fan and Li, Maokun and Hou, Yangqing and Jiang, Shuidong and Liu, Lei},
  journal=TAP, 
  title={A Broadband High-Efficiency Reconfigurable Reflectarray Antenna Using Mechanically Rotational Elements}, 
  year={2017},
  volume={65},
  number={8},
  pages={3959-3966},
  month=aug,
  doi={10.1109/TAP.2017.2708079}}

@ARTICLE{9091906,
  author={Miao, Yang and Haneda, Katsuyuki and Naganawa, Jun-Ichi and Kim, Minseok and Takada, Jun-Ichi},
  journal=TWC, 
  title={Measurement-Based Analysis and Modeling of Multimode Channel Behaviors in Spherical Vector Wave Domain}, 
  year={2020},
  volume={19},
  number={8},
  month={May},
  pages={5345-5358},
  doi={10.1109/TWC.2020.2992533}}

@ARTICLE{5159555,
  author={Alayon Glazunov, AndrÉs and Gustafsson, Mats and Molisch, Andreas F. and Tufvesson, Fredrik and Kristensson, Gerhard},
  journal=TAP, 
  title={Spherical Vector Wave Expansion of {G}aussian Electromagnetic Fields for Antenna-Channel Interaction Analysis}, 
  year={2009},
  volume={57},
  number={7},
  pages={2055-2067},
  month={Jul.},
  doi={10.1109/TAP.2009.2016686}}

\end{document}